\documentclass{article}
\usepackage[utf8]{inputenc}
\usepackage{amsmath}
\usepackage{amsthm}
\usepackage{amssymb}
\usepackage{babel}
\usepackage{color}

\usepackage{xcolor}

\newtheorem{theo}{Theorem}[section]  
\newtheorem{lem}[theo]{Lemma}

\newtheorem{coro}[theo]{Corollary}
\newtheorem{fa}[theo]{Fact}

\theoremstyle{definition}
\newtheorem{defi}[theo]{Definition}
\newtheorem{ex}[theo]{Example}

\newenvironment{myfont}{}{}

\newcommand{\n}{\noindent}

\newcommand{\Var}{\textsf{Var}}
\newcommand{\For}{\textsf{For}}
\newcommand{\Ex}{\textsf{Ex}}
\newcommand{\Con}{\textsf{Con}}

\newcommand{\N}{\mathbb{N}}

\def\var{\textsf{var}}
\def\Pow{\mathcal{P}}
\def\R{\textsf{R}}

\def\LV\mathbb{LV}

\def\s{\mathfrak{S}}
\def\m{\mathfrak{M}}

\def\S{\mathcal{S}}
\def\M{\mathcal{M}}

\def\Si{\mathcal{S}^{k}}
\def\Mi{\mathcal{M}^{k}}
\def\Mn{\mathcal{M}^{n}}

\def\GS{\mathcal{GS}}
\def\TR{\mathcal{TR}}

\newcommand{\TE}{\textsf{Lab}}

\def\LV{\mathbb{LV}}

\def\K{\textsf{K}}

\def\t{\textsf{t}}
\def\f{\textsf{f}}

\def\SEp{\textbf{S}}

\title{Tableau methodology for propositional logics}
\author{T. Jarmużek\footnote{Nicolaus Copernicus University in Toru\'n}, R. Gor\'{e} 
 \footnote{{Monash University, Melbourne, Australia}}}
\date{}

\begin{document}

\maketitle

\begin{abstract} We set out a general methodology for producing tableau systems for propositional logics via a tableau metatheory
  that provides general and formal notions for different tableau systems that vary by semantics or formulae. Moreover, by dint of
  these general notions, some facts, independent of their applications to a particular propositional logic, can be proved. One of
  the examples is the tableau metatheorem that simplifies the process of constructing a complete tableau system for a given logic,
  just reducing it to checking specific properties of the tableau rules within the analyzed, particular system.
  In our paper we generalize  an abstract consistency property proposed by R. Smullyan and M. Fitting from the modal case to the others.
  Such a methodology is essential for a deeper and
  universal treatment of tableau methods for various propositional languages and semantics.

\end{abstract}

\textbf{Keywords:} metatheory, methodology of logic,  methodology of tableaux, proof theory, propositional logic, tableau methods.

\section*{Introduction and overview}\label{sec:zero}

We explore the relationships between widely understood propositional logic and their  tableau proofs. We set out a
methodology for constructing adequate tableau systems for different sorts of propositional logic that can be determined with
generalized relational semantics.\footnote{A  similar strategy, but relative to the context of syllogistic logic was implemented by us~\cite{TJiRG2021} where the tableau metatheory was developed for syllogistic languages.}

The initial and basic idea that motivates our approach is the following observation: the syntax and semantics of a given logic always determine a minimal syntax and structure of a tableau system for the logic along with other properties. So, it seems reasonable to propose some metatheory for a class of logics that are very similar from  a syntactical and semantic viewpoint. 
One can naturally ask: how can  we benefit from such a theory? We show how the metatheory we propose  makes the process of defining adequate tableau systems much easier. Since we define very general notions that cover all tableau issues concerning a very wide class of propositonal logics and demonstrate crucial relationships between them, we can apply them automatically in particular cases,  concentrating only on remaining details, specific to the~particular propositional logic under examination.\footnote{The ideas behind the presented approach were outlined relative to certain contexts in \cite{TJ2007}, \cite{TJ2008}, \cite{TJ2013}, \cite{TJ2013b}. They have been developed and improved here to cope with propositional logic.} Additional, and probably more important advantages of this metatheory that open future interesting  research areas in tableau-proofs-approach are listed in the last section  (Section \ref{sec:6}).

At the beginning let us notice that in the theory of tableaux, we can distinguish three kinds of approaches  to their construction.

First, tableau proofs are either with a signed or unsigned language. \emph{Signed} means that in a language of tableau proofs, additional symbols to denote logical values are used, while \emph{unsigned} means that in the tableau proofs we do not use any such symbols for logical values. This is a traditional division. It is worth mentioning that the first appearance of the tableau method in Beth~\cite{EB1955} used signed tableaux (\cite{ZL1960}).

Second, tableaux either can be built with  languages that contain labels that denote possible worlds (points of relativization), or can be built with languages without labels (see \cite{MF1983}, \cite{RG1999}).

Third,  a way of construction of  tableaux (and branches)  can be divided into \emph{nodes based on formulae} or  \emph{nodes based on sets of formulae}. While the former seems to be  usually of didactic form (this approach  is for example extensively outlined in \cite{GP2008}), the latter is more paradigmatic (see for example \cite{MF1983} and \cite{RG1999}) and has a strong connection to sequent calculi.

In our paper we set out a generalization of  all these  aspects of tableau theory for propositional logic that can be tweaked to any of the more specific forms mentioned above.  We obtain the required generality by using \emph{generalized labels} (but we call them \emph{labels}). \label{Generalized label} The generalized labels can obviously code points of relativization, but also other important semantic (and not only semantic) aspects, such as logical values, an object-property  of belonging/non-belonging to a denotation of a given term or possibly other things.\footnote{Other possible uses of generalized labels can be: (a) tracking the origin of decomposed formulae (see for example case for relevance tableaux \cite{DM2017} or paraconsistent tableaux \cite{TJiMT2013}, \cite{TJiMT2015}), (b) quasi-negation (-)/quasi-assertion (+) in many valued tableaux or FDE tableaux, including Routley Star (see for example \cite{GP2008}).} Moreover, since we are developing a metatheory, we should be open to new roles for labels, which  may only emerge when we use the metatheory for particular cases of new propositional logics.

Last but not least, our approach is of the nodes-based-on-sets-of-formulae kind but different to others of this kind in at least two ways. First,  decomposed expressions are collected rather than deleted, so at any stage of a tableau branch,  the full
information on the proof is still available. So additional constraints on tableau rules can be imposed directly on the inputs of rules (since in our approach tableau rules are sets of n-tuples of sets such that the input is always a proper subset of the output). Second, in fact we do not use direct tree structures with nodes, since branches are strictly monotonic sequences of input/output sets, while tableaux are sets of such branches (selected with some additional conditions). However,  the~remaining approaches (with nodes based on formulae as well as  sets of formulae) can be defined using our apparatus, since our approach is more abstract.

In Section \ref{sec:1} we present a language  that includes what is common to any propositional logic, whilst not excluding  richer grammatical constructions, since we want our theory to cover  all possible propositional languages.

In Section \ref{sec:2} we introduce general semantic structures for propositional logics, as well as notions of satisfiability which together form models. It turns out that models semantically determine particular propositional logics, which is illustrated by examples. In keeping with the spirit of generality, in Section  \ref{sec:3} we propose a syntax for our tableau
language to describe properties of general semantic structures in tableau proofs. Finally, we introduce a notion of a set of satisfied expressions.

Section  \ref{sec:4} is completely devoted to problems of what is generally a tableau rule, a branch, a tableau etc. A multistage, set-theoretical  construction of these notions is proposed. As a novelty in the field of propositional logic a branch consequence relation as the
tableau counterpart of a semantic consequence relation for a given logic is also introduced. 

In Section \ref{sec:5} all important relationships between these notions are proved to present  general connections between the tableau notions  and general semantics. 
There we also generalize to some extent the aforementioned notion abstract consistency property. This establishes  sufficient conditions for a complete and sound tableau system that are in this section formulated.

In the last section of the paper (Section \ref{sec:6})  we present  an application of the metatheory to the construction of adequate tableau system for a bi-modal, three-valued logic, and we outline some perspectives for the further development of the tableau metatheory.

\section{Propositional languages}\label{sec:1}

In order to explore the general view on tableaux for propositional languages we need to define the notion of propositional language generally. 

Let $\N$ be the set of natural numbers. The propositional logics are defined in the languages that are built with:  variables $\Var  = \{p_{i} |  i \in \mathbb{N} \}$ (in practice we will use: $p$, $q$, $r$ , propositional connectives $\Con_{I}^{K}  = \{ c_{i}^{k} |  i \in I, k \in K \}$, where $I, K \subseteq \N$ and parentheses $)$, $($. In case of the connectives $c_{i}^{k}$ the subscript $k$ is an arity, while $i$ is just an index
for the connective. What is the purpose of $I$ and $K$? Could we not just say  $\Con  = \{ c_{i}^{k}: i, k \in \N\}$? Would we ever want to define two sets of connectives such as $\Con_{1,2,3}^{2}$ and $\Con_{4,5,6}^{2}$?  Our approach is to be general, so we consider $\Con_{I}^{K}$ for some $K$, $I$.

For the sake of our considerations, however,  we assume $\Var$, brackets and $\Con_{K}^{I}$, for some $K$, $I$ (rather non-empty), writing then just $\Con$. On this ground we have a standard definition of set of formulae.

\begin{defi}[Formula]\label{Formula} Set of formulae is the least set $\For$ that satisfies the following conditions:

\begin{enumerate}

\item $\Var \subseteq \For$

\item $c_{i}^{k} (A_{1}, \dots, A_{k}) \in \For$, if $A_{1}, \dots , A_{k} \in \For$ and $ c_{i}^{k} \in  \Con$, for some $i, k \in \N$.

\end{enumerate}

\end{defi}

\section{General semantics for propositional languages} \label{sec:2}

A language in which tableau proofs can be carried must relate  formulae to some aspects of semantic structures, so the general theory of propositional  tableaux needs a universal notion of semantic structure, and finally model for a propositional logic.

\begin{defi}[General semantic structure]\label{General semantic structure}

\textit{A general semantic structure } is the following  ordered triple:  $\langle \{W_{i}\}_{i\in M}$, $\{ R_{j}\}_{j\in N}$, $\vartheta \rangle$, where: 

\begin{itemize}

\item[(a)] $M$, $N$ are sets of indices (and $M \cap N = \emptyset$),

\item[(b)]  $\{W_{i}\}_{i\in M}$ is a non-empty family of indexed, non-empty sets (called \emph{set of domains}),

\item[(c)] $\{R_{j}\}_{j\in N}$ is a possibly empty family of relations,

\item[(d)] $\vartheta$ is a valuation of variables in~all domains, so  $\vartheta \colon \bigcup_{i \in M} W_{i} \times \Var \longrightarrow \{0, 1 \}$.

\end{itemize}

 The set of all general semantic structures is denoted by $\GS$.  They can be also called \emph{semantic structures} or just \emph{structures}.

\end{defi}

Some convention can be adopted. Let $\s = \langle \{W_{i}\}_{i\in M}, \{ R_{j}\}_{j\in N}, \vartheta \rangle$. For  all $i \in M$, $W_{i}^{\s} = W_{i}$; for all  $j \in N$, $R_{j}^{\s} = R_{j}$,  and $\vartheta^{\s} = \vartheta$. The same convention is assumed for all structures interpreted later as models.

Some short comments on definition~\ref{General semantic structure} are needed. A structure  $\langle \{W_{i}\}_{i\in M}, \allowbreak  \{ R_{j}\}_{j\in N}, \allowbreak
\vartheta \rangle$ is an ordered triple where:
\begin{enumerate}
\item $\{W_{i}\}_{i\in M}$ is a non-empty set of non-empty domains. The domains  can play at least three roles: points of
  relativization (sometimes they are called \emph{possible worlds} or just \emph{worlds}); denotations for sentences (see \cite{TJiMK2022}); and logical values. But they can have more applications.  Since $\{W_{i}\}_{i\in M}$ is non-empty, we always have a non-empty set of
  points of relativization.  Clearly, even in non-modal logics,  we have a world, but when there is only one, it is usually omitted.

\item $\{R_{j}\}_{j\in N}$ includes different relations of various arities. They enable the modelling of modalities in intensional contexts. They also can  be  additional components (e.g. the hereditary relation $\sqsubseteq$ or ternary relations for relevant models) or functions (like Routley Star (see \cite{RRiVR1972})), and so on. However, $\{R_{j}\}_{j\in N}$ can be empty.

\item
Finally, we have $\vartheta$, a valuation of variables in~all domains. It describes a relationship between any propositional variable  $p \in \Var$
and any object $a$ in $W_{i}$, where $i \in M$. For example if $W_{i}$ is a set of possible worlds and $a \in W_{i}$, $\vartheta(a, p) = 1$
can mean that $p$ is true at $a$; if $W_{i}$ is a Cartesian product of possible worlds and logical values, then $a \in W_{i}$ is a pair $\langle
b, l \rangle$ (where $b$ is a possible world and $l$ is a logical value) and  $\vartheta(a, p) = \vartheta(\langle b, l\rangle, p) = 1$ can in the
metalanguage mean that the sentence  ``\emph{$p$ is at $b$}'' (or that ``\emph{$p$ happens at $b$}'', etc.) has logical  value $l$.  Clearly, the function in many situations can be excessive. But this is not a problem. By dint of the bi-valued function $\vartheta$ we can reduce even complex semantic properties to basic binary decisions: either  yes or no. However, before we complicate the situation, let us consider a very simple example to explore some initial intuitions.

\end{enumerate}

\begin{ex}\begin{myfont} [A variant of logic $\SEp$]\label{Przyklad podlogiki S} Here we define a variant of logic of content relationship $\SEp$ (see: \cite{TJiMK2020}). The idea behind the content relationship logic is that sentences have ascribed not only logical values, but also sets that represent abstract notion of content. Let our  language be built with $\Var$ and  three connectives $\neg, \wedge, \rightarrow$.  We consider a set $\S$ of semantic structures  that contains only structures $\s = \langle \{W_{i}\}_{i\in M}, \{ R_{j}\}_{j\in N}, \vartheta \rangle$, where $M = \{ 1, 2 \}$, $W_{1}  = \{w \}$ (one
  point of relativization or one world), $W_{2}$ is an arbitrary non-empty set of non-empty sets, $\{ R_{j}\}_{j\in N}$ is empty, $\vartheta$ is a valuation of variables in domains $\vartheta \colon \bigcup_{i \in M}
  W_{i} \times \Var \longrightarrow  \{0, 1 \}$ that  for all $p\in \Var$  satisfies the condition:

\begin{itemize}

\item[] (cont) there exists  exactly one $U \in W_{2}$  such that $\vartheta(U,x) = 1.$

\end{itemize}

The condition (cont) is to guarantee that no variable has an empty content.\footnote{The condition is because in the logics of
  content relationship it is usually assumed that the content of sentences is
  non-empty.} 
Could the structures not be simplified to $\langle W_{2}, \vartheta\rangle$, since we have no relations and only one point of
relativization, meaning that they can be omitted? The answer is no, because the case under study is extremely simple. In general
we must stick close to a more general pattern of structures, to allow for different possibilities. In fact, we have already
simplified things by leaving out domains with additional values and designated values since we are defining a two-valued
logic. But we will add these in example \ref{Example mulitivalued and intensional}.

So far, we have the definition of structures. But what makes a structure a model for the language? The answer is: a satisfiability relation that stipulates the conditions under which  sentences are true or false.

Before we introduce the relation, let us define some notation. By $\Pow$ we mean the powerset operation, while by $\var \colon \For \longrightarrow \Pow(\Var)$ the function that collects all variables contained in a given formula.

For any formula $A$ and for all $\s \in \S$, and for $\vartheta^{\s}$, we define an intended content of a formula in the model as:

\[ s^{\s}(A) =  \bigcap \{U \in W_{2}^{\s} | \text{ for some } p  \in \var(A),  \vartheta^{\s}(U, p) = 1 \}. \footnote{In the original version of $\SEp$, the content of formula in a model is defined as the union of the content of variables: $s^{\s}(A) =  \bigcup \{U \in W_{2}^{\s} | \text{ for some } p \in \var(A),  v^{\s}(U, p) = 1 \}$, which is called \textit{union set assignment}. Here (and in the example \ref{Example mulitivalued and intensional}) we employed intersection in place of union, which we can call \textit{intersected  set assignment}. It results in  a weaker logic than $\SEp$.} \]

We introduce now the relation $\models$ between the structure
$\s$, points of relativization $W_{1}^{\s}$ and formulae $\For$ (if the relation $\models$ does not hold, we write $\not \models$) --- this is the satisfiability relation:


\begin{itemize}

\item[(var)]  $\s, w \models p$ iff $ \vartheta( w,  p )=1$, if $p \in \Var$

\item[($\neg$)]  $\s, w \models \neg A$ iff $\s, w \not \models A$

\item[($\wedge$)]    $\s, w \models  A \wedge B$ iff  $\s, w \models A$ and  $\s, w \models B$

 \item[($\rightarrow$)]  $\s, w \models A \rightarrow B$ iff for all $A, B \in \For$, all $\s \in \S$ and $w \in W_{1}^{\s}$:

 \begin{enumerate}
     \item  $\s, w \not \models A$ or $\s, w \models B$
     \item $s^{\s}(A) \cap s^{\s}(B)  \neq \emptyset$. 
     
 \end{enumerate}

\end{itemize}

\noindent So, we have in fact defined  $\models$ for $W_{1}^{\s}$.  Moreover, we could reduce $\s, w \models A$ to $\s \models A$, since $W_{1}^{\s}$ is a~singleton in this example. However, as we have said, we keep in mind the more general pattern, allowing more worlds.

Finally, in the standard way we define a consequence relation as the relation  $\models_{\S} \subseteq \Pow(\For) \times \For$ determined by preserving truth in the class of all structures $\S$ under conditions (var), ($\neg$), ($\wedge$), ($\rightarrow$):
\[ X \models_{\S} A \text{ iff } \forall \s \in \S, w \in W_{1}^{\s} ((\forall B \in X,~ \s, w \models B) \Rightarrow \s, w
\models A ). \]


\noindent In this way, we obtain a  semantically determined logic that is a sublogic of content related logic $\SEp$. The logic $\SEp$ was introduced in \cite{RE1990} and examined by tableau approach, among others, in \cite{TJiMK2022}. 

For example, when we take $X = \{ p \}$ and $A = p  \rightarrow \neg (\neg p \wedge \neg q)$, then we have $\{p\} \not \models_{\S} p  \rightarrow \neg (\neg p \wedge \neg q)$. However, when we assume  $X = \{\neg (\neg (p \rightarrow q) \wedge \neg (q \rightarrow p)) \}$ and $A = (p \wedge (p \rightarrow q)) \rightarrow q$, then we have $\{ \neg (\neg (p \rightarrow q) \wedge \neg (q \rightarrow p)) \}  \models_{\S} (p \wedge (p \rightarrow q)) \rightarrow q$, by the conditions (var), ($\neg$), ($\wedge$), ($\rightarrow$).

At the end of the example it is worth noting that the defined structures can be reduced just to ordered pairs $\langle \vartheta, s \rangle$, where $\vartheta \colon \Var \longrightarrow \{0, 1 \}$,  $s \colon \For  \longrightarrow W$, and $W$ is a non-empty set of  non-empty sets and for complex formula $A$:

\[ s(A) =  \bigcap \{U \in W  | \text{ for some } p \in \var(A),  s(p) = U \}. \]


\n Under the truth conditions we have introduced in this example, the structures $\langle \vartheta, s \rangle$ become models for the sublogic $\SEp$. This shows that the general structures we introduce can capture such cases. 
Later in this section we will present one more example. \end{myfont}
\qed{}
\end{ex}

We next introduce some further notions that will be of use.  Let us consider some subset of structures $\S \subseteq \GS$. In
order to interpret $\S$ as a set of models we assume that there exists an object $k$ that for all $\s \in \S$, if
$\s = \langle \{W_{i}\}_{i\in M}$, $\{ R_{j}\}_{j\in N}$, $\vartheta \rangle$, then $k \in M$. In other words, to become models, all
structures in the set $\S$ are required to share in the sets indexing their families of domains $\{W_{i}\}_{i\in M}$ at least one
index $k \in M$. If so, then all of the structures include in $\{W_{i}\}_{i\in M}$ a non-empty set $W_{k}$ with the same index
$k$. The sets indexed with $k$ are intended to be treated in models from $\S$ as sets of points of relativization (particularly as
one point of relativization in any structure). Formally:

\begin{defi}[k$^{th}$ domain structure]\label{kth domain} Let $\s = \langle \{W_{i}\}_{i\in M},\{ R_{j}\}_{j\in N}, \vartheta \rangle$.
  $\s$ is called \emph{a structure with the k$^{th}$ domain} iff $k \in M$.  Let $\S$ be a set of structures.  By $\Si$ we
  stipulate that all structures in $\S$ contain one common k$^{th}$ domain, for some $k$.
\end{defi}

For instance, in example \ref{Przyklad podlogiki S} above, domain $W_{1}$ always serves as k$^{th}$ domain, with $k=1$.

A set of structures $\S \subseteq \GS$ becomes a set of models of $\For$ only if some special relation between $\S$ and $\For$ is
set down. It is usually a~ternary relation of satisfiability between structure $\s$, a point of relativization $w$ and any formula
$A \in \For$: if all structures in $\S$ have only one relativization point, the relation can be treated as binary. In case of semantically
determined logics this relation is often given inductively, by some truth conditions, as in example~\ref{Przyklad podlogiki
  S}. However, here we do not have any particular truth conditions as we are working on a very abstract level. So we must also use
some general notion of being a model of $\For$, without specifying any truth conditions.

First, we introduce some useful notation. Let $\s = \langle \{W_{i}\}_{i\in M},$ $\{ R_{j}\}_{j\in N}$, $\vartheta\rangle \in \GS$. Then
$\overline{\s} = \{Q_{\s} \colon Q_{\s} \subseteq (\bigcup_{i \in M} W_{i}) \times \For \}$ is the set of all subsets of
$(\bigcup_{i \in M} W_{i}) \times \For$, so just $\Pow((\bigcup_{i \in M} W_{i}) \times \For)$.  For any particular
$Q_{\s}$, the expression $w Q_{\s} A$ intuitively means, that formula $A$ is satisfied in structure $\s$ at point $w$. We see that
there are many possibilities for how to determine $Q_{\s}$ (set $\overline{\s}$ contains all of them for $\s$), but always it is
defined uniquely. When we choose some $Q_{\s}$ it can be also excessive in many features. But generally this does not cause any
trouble, since it is clear when a formula is satisfied at a relativization point from $W_{k}$.

\begin{defi}[Models of \For]\label{Models of For} Let $\S \subseteq \GS$. Let $Q \subseteq \bigcup \{\bigcup_{i \in M} (W_{i}^{ \s}) \times \For | \s \in  \S \}$. 
  Then, $\S$ are \emph{models of }$\textsf{For}$ \emph{with respect to} $Q$.
\end{defi}

\n Definition \ref{Models of For} expresses the fundamental idea that being a model of a~given set of formulae  provides a relation between particular parts of domains, especially points of relativization and formulae. In other words, $\s$ is a model of \For,  if we can define (usually, inductively, if it is computable) a relation of \emph{being satisfied} (and by
complementation,  of \emph{being not satisfied})  in model $\s$  at any $w$ that belongs to $\bigcup_{i \in M} W_{i}$. It is obvious
that set $\S$ can be the base for different models if $Q$ varies. We now make certain assumptions for the remainder of the paper:

\begin{itemize}

\item for $\For$ (defined in \ref{Formula}) we assume a set of semantic structures $\S \subseteq \GS$

\item structures $\S$ are with the k$^{th}$ domain, for some $k$, so we write $\Si$


\item relation between $\Si$ and $\For$ is established by some satisfiability relation.

\end{itemize}

\noindent Since $\Si$ are models of $\For$ with respect to  satisfiability  by definition~\ref{Models of For}, we will write $\Mi$ instead of $\Si$.
We are now ready to define  in a standard way a semantic consequence relation. We will use an abbreviation. Let $X \subseteq \For$, $\m \in \Mn$ and $w$ be a relativization point in k$^{th}$ domain in $\m$ ($w \in W_{k}^{\m}$). Then: $\m, w Q A$ means that in model $\m$ formula $A$ is satisfied at the point $w$; while $\m, w Q X$ means that for all $B \in X$: $\mathfrak{M}, w Q B$.

\begin{defi}[Semantic consequence relation]\label{Semantic consegence relation} Let~$\Mn$ be a set of models of $\For$ with respect to relation $Q$, and with n$^{th}$  domain. Let $X \subseteq \textsf{For}$, $A \in \textsf{For}$.



\begin{itemize}

\item $(\maltese)$ \label{Pattern of local semantic consequence relation} Formula $A$ is  \textit{a local semantic consequence of set of formulae} $X$ \textit{with respect to a set of
    models} $\Mn$ (in short:  $X \models_{\Mn} A$) iff for all $\mathfrak{M}= \langle \{W_{i}\}_{i\in M}, \{ R_{j}\}_{j\in N}, \vartheta \rangle \in \Mn$, for all $w \in W_{n}^{\m}$:

\medskip

 if $\mathfrak{M}, w Q X$, then $\mathfrak{M}, w Q A$.

\item Formula $A$ is  \textit{a global semantic consequence of a set of formulae} $X$ \textit{with respect to a set of models} $\Mn$ (in   short:  $X \models^{g}_{\Mn} A$) iff for all $\mathfrak{M}= \langle \{W_{i}\}_{i\in M}, \{ R_{j}\}_{j\in N}, \vartheta \rangle \in \Mn$:

\medskip

 if for  all $w \in W_{n}^{\m}$: $\mathfrak{M}, w Q X$, then for  all $w \in W_{n}^{\m}$: $\mathfrak{M}, w Q A$.

\end{itemize}

\end{defi}

Clearly, it is always the case that  $\models_{\Mn}\;\subseteq\;\models^{g}_{\Mn}$. So now, we have a good starting point for investigating
the relationships between $\models_{\Mn}$, $\models^{g}_{\Mn}$ and tableau systems that can be constructed for both. However, we
concentrate only on local consequence relation, leaving the other for another occasion.


Hence, since we  have $\For$ and models of $\For$, by definition~\ref{Semantic consegence relation} we automatically get the
logic defined semantically on $\For$ by the set of models $\Mi$ that is identified with relation  $\models_{\Mi}$. To simplify notation, later we will write just $\M$ and $\models$, remembering that $\models$ is defined by models from $\M$ on worlds from the  k$^{th}$ domain of any model in $\M$. \label{Semantic assumptions}  We shall use $\models$ also to denote the satisfiability relation between models, worlds, formulae, and possibly other semantic entities. These assumptions we denote by $(\maltese\, \maltese)$.\label{Local semantic consequence relation}

Let us consider a last example of the application of our general semantic notions. It is the application to the case of many-valuedness and intensionality for a language. It is worth noticing that the complications in  our formulation arise mainly because of the complex structures of different possible propositional logics that we cover in our metatheory.

\begin{ex}\begin{myfont}[A many-valued and intensional logic with a content--related implication] \label{Example mulitivalued and intensional} We assume a language built with 
negation $\neg$, conjunction $\wedge$, implication $\rightarrow$, and a necessity operator $\Box$. 
We would like to give an intensional and many-valued (with $m$ logical values and $n$ designated) interpretation of the language. For this we choose set $\S \subseteq \GS$  of structures  $\langle \{W_{i}\}_{i\in M}, \{ R_{j}\}_{j\in N}, \vartheta \rangle$, where:

\begin{itemize}

\item $M = \{ 1, 2, 3, 4, 5 \}$; 

\item  $W_{1}$ is a non-empty set of points of relativization;

\item for $m \in \mathbb{N}$, the set $W_{2} = \{l^{1}, \dots, l^{m} \}$ is a
  set of $m \geq 2$ logical values,  while for $n \in \mathbb{N}$, the set $W_{3} = \{l^{k_{1}}, \dots, l^{k_{n}} \}$ is a   set of $n$ designated logical values, with $W_3 \subset W_2$ and $1 \leq k_{1}, \dots, k_{n} \leq m$ (For ease we will later write  $W_{V}$ and $W_{DV}$,  instead of $W_{2}$, and respectively $W_{3}$, but keep in mind that they are the 2$^{nd}$ and the 3$^{rd}$  domains of $\S$.);


\item $W_{4} = \{ W_{w}\}_{w\in W_{1}}$ is a family of non-empty sets of sets indexed by the relativization points from $W_{1}$;

\item $W_{5} = \{ \langle w, l \rangle \mid w \in W_{1}, l \in W_{V} \}$ is a set of pairs: a world and a logical value;


\item $\{ R_{j}\}_{j\in N} = \{ R_{1}, R_{2} \}$, (so $N  = \{ 1, 2 \}$)  is a doubleton  with:

\begin{itemize}

\item $R_{1} \subseteq W_{V} \times W_{V}$ is an order on the set of logical values; we put linear order: for all
  $l^i, l^j \in \{l^{1}, \dots, l^{m} \}$, $l^i R_{1} l^j$ iff $i \leq j$, as $i, j \in \mathbb{N}$; of course, for any non-empty
  $W' \subseteq W_{V}$ in respect of $R_{1}$ there exists $min(W')$ as well as $max(W')$; intuitively $max(W_{V})$ we interpret as
  \emph{truth}, while $min(W_{V})$ as \emph{falsity}, so we additionally assume that $max(W_{V}) \in W_{DV}$, whereas
  $min(W_{V}) \not \in W_{DV}$,

\item $R_{2} \subseteq W_{1} \times W_{1}$ is  just an accessibility relation for interpretation of $\Box$-formulae,

\end{itemize}

\item valuation $\vartheta \colon \bigcup_{i \in M} W_{i} \times \Var \longrightarrow  \{0, 1 \}$, that for all $p \in \Var$ satisfies the
  conditions:
\end{itemize}

\begin{itemize} 

\item[(a)] for all $p \in \Var$ and all $w \in W_{1}$ there exists exactly  one $U \in W_{w} \in \{ W_{w}\}_{w\in W_{1}}$, such that  $\vartheta( U, p) = 1$, 

\item[(b)] for all $p \in \Var$ and all  $w\in W_{1}$ there exists exactly one $l \in W_{V}$,  such that  $\vartheta(\langle w,  l \rangle, p ) = 1$, \hfill (clearly, $\langle w,   l \rangle \in W_{5}$)

\item[(c)]  for all $p \in \Var$, $y \in \bigcup_{i \in M \setminus \{ 4, 5 \}} W_{i}$,  $\vartheta(y, p ) = 0$.

\end{itemize}

Condition (a) defines what is a content of a variable $p$ in a given world $w$. As we can see, thanks to $W_{4} = \{ W_{w}\}_{w\in W_{1}}$, the content can vary from world to world, since with any world there is associated a set of sets that represent content of atomic sentences. Condition (b) is important for the construction because (b) determines the relationship between a possible world $w$, a variable $p$, and a logical value: in world $w$ sentence $p$  has exactly one logical value from $W_{V}$. Intuitively, such sentences describe basic semantic facts in the structures. Condition (c) just fulfils the requirements that the function must be defined for all domains.

Before we introduce the relation, let us define some notation. We generalize into possible worlds the notion of content that we introduced in the example  \ref{Przyklad podlogiki S}. For any formula $A$ and for all $\s \in \S$, $w \in W_1^{\s}$ and $\vartheta^{\s}$, we define 

\[ s^{\s}_{w}(A) =  \bigcap \{U \in W_{w}^{\s} \in W_{4}^{\s} | \text{ for some } p \in \var(A),  \vartheta^{\s}(U, p) = 1 \}. \]

Now we define a satisfiability relation $\models$ in two steps. Let us take $\s \in \S$, where $\s = \langle \{W_{i}\}_{i\in M}, \{ R_{j}\}_{j\in N}, \vartheta \rangle$.  We would like to know under what conditions $\s, w \models A$, where $w \in W_{1}$ and $A$ is a formula. The fundamental idea is that for each $w$, a formula $A$ should have exactly one value $l$ from $W_{V}$. Intuitively, if $l \in W_{DV}$ (i.e.\ is a designated value), then $A$ is satisfied in $w$ (i.e.\ $\s, w\models A$); if $l \not \in W_{DV}$ (i.e.\ is not a designated value), then $A$ is not satisfied in $w$ (i.e.\ $\s, w \not \models A$).

Our definition is inductive, so we  will define atomic cases directly, and make an inductive step. Let $p \in \Var$. By the definition of $\s \in \S$ for  all  $w\in W_{1}$ there exists exactly one $l \in W_{V}$,  such that  $\vartheta(\langle w,  l \rangle, p ) = 1$. So, we define

\[\s, \langle w, l \rangle \models p\text{ iff } \vartheta(\langle w,  l \rangle, p ) = 1\]

\n and finally

 \[\s, w  \models p \textit{ iff } \exists l \in W_{V} (\s, \langle w, l \rangle \models p \, \& \, l \in W_{DV}).\]

We consider cases of possible  formulae: $\neg A$, $A \wedge B$, $A \rightarrow B$, $\Box A$, where  $A, B$ are formulae,  and the remaining symbols are the assumed connectives. 
These cases are more complex, so we make an inductive assumption. Let us consider  formulae $A$, $B$ such that for any of
them and  all $w \in W_{1}$ there exists such a unique $l_{1} \in W_{V}$ that $\s, \langle w, l_{1} \rangle \models A$  and there
exists such a unique $l_{2} \in W_{V}$ that $\s, \langle w, l_{2} \rangle \models B$.

We introduce a meaning for $\neg$. 

\[\s, \langle w, l \rangle \models \neg A \text{ iff } l =  \begin{cases} min(W_{V}), &\text{if } l_{1} = max (W_{V}) \\ max (W_{V}), &\text{if } l_{1} = min (W_{V}) \\ l_{1}, &\text{otherwise}. \end{cases}\]

\n So,   truth is a complement  of  falsity, and vice versa, but other values are a complement for themselves. Finally,

 \[\s, w  \models \neg A \textit{ iff } \exists l \in W_{V} (\s, \langle w, l \rangle \models \neg A \, \& \, l \in W_{DV}).\]

For the case of $\wedge$ we define: if   $\s, \langle w, l_{1} \rangle \models A$  and  $\s, \langle w, l_{2} \rangle \models B$,
then $\s, \langle w, l  \rangle \models A \wedge B$ iff $l = min (\{ l_{1}, l_{2} \})$ --- hence, it is also a standard approach
--- and $\s, w  \models  A \wedge B$ iff $\exists l \in W_{V} (\s, \langle w, l \rangle \models A \wedge B  \, \& \, l \in W_{DV})$.

Now, we define the condition for the content-related implication $\rightarrow$. 

\[\s, \langle w, l \rangle \models  A \rightarrow B \text{ iff } l =   \begin{cases}  l_{2}, &\text{if }  s^{\s}_{w}(A) \cap s^{\s}_{w}(B) \neq \emptyset \text{ and }  l_{1} R_{1} l_{2} \\ l_{1}, &\text{if }  s^{\s}_{w}(A) \cap s^{\s}_{w}(B) \neq \emptyset \text{ and }  l_{2}  R_{1} l_{1} \\ min(W_{V}), &\text{otherwise}; \end{cases}\]

\n and:

 \[\s, w  \models  A \rightarrow B \textit{ iff } \exists l \in W_{V} (\s, \langle w, l \rangle \models  A \rightarrow B \, \& \, l \in W_{DV}).\]

Finally, we take into account intensionality, the case of $\Box$. Having in mind a possible world $w$, we consider a set of  all
worlds that are accessible from $w$, so $R_{w} = \{ w' \in W_{1} \mid w R_{2} w' \}$. We now define a set of logical values that formula $A$ takes in particular worlds from $R_{w}$.
It is the set
 $R_{w}^{l} = \{l \in W_{V} \mid \exists w' \in R_{w} \text{ such that }  (\s, \langle w', l \rangle \models A) \}$. Then, we propose that:

 \[ \s, \langle w, l \rangle \models \Box A \textit{ iff } l =\begin{cases} min(R_{w}^{l}), &\text{if } R_{w}^{l} \text{ is non-empty} \\ max (W_{V}), &\text{otherwise}. \end{cases}\]
 

\noindent At the end,  we put $\s, w \models \Box A$ iff
$ \exists l \in W_{V} (\s, \langle w, l \rangle  \models \Box A \, \& \, l \in W_{DV})$.

As we see,  for any formula $A$ in all structures $\s \in \S$ and all possible world $w \in W_{1}$ we can determine whether
$\s, w   \models A$, or $\s, w  \not \models A$. Therefore, we immediately have semantically determined models on
structures $\S$ and  --- by condition $(\maltese)$ (def. \ref{Semantic consegence relation}) --- a modal, many--valued  logic of a content-related implication of some kind. \end{myfont}\end{ex}

 Note that  \ref{Example mulitivalued and intensional} is only an example of how to define a propositional logic with a subset of structures of $\GS$. It is probable that all kinds of semantically determinable propositional logics   can be determined  by a subset
of structures $\GS$ when we suitably introduce satisfiability  relation between structures,  possible worlds in their k$^{th}$  domains and formulae. This is the range of what is covered by our tableau metatheory for propositional logics we shall introduce in the subsequent sections.

One can ask: how separate the wheat from the chaff? How will you separate ``good'' logics from ``silly'' ones?  Our theory --- as all theories --- is conditional. If you have a good, semantically determined propositional logic that satisfies certain (given later) conditions, you will get an adequate tableu system. If you put in the the high-quality wheat, you will pull out the good bread.

\section{Fundamental tableau theory notions}\label{sec:3}

To describe the properties of structures $\M$ for language $\For$ (that we assumed in Section~\ref{Semantic assumptions}) in the
tableau metatheory we need some syntactic apparatus.  Firstly, we list symbols:

  \begin{description}
  \item[\rm{Indices:}] a set $\mathbb{I}$ of indexes (equal to the set $\mathbb{N}$ of natural numbers for simplicity)
  \item[\rm{Function Symbols:}] a set  of n--ary functional symbols $w^n_i$ where $n \geq 1$ and $i \geq 1$:
    $w_{1}^{1}$, $w_{2}^{1}$, $w_{3}^{1}$, \dots, $w_{1}^{2}$, $w_{2}^{2}$, $w_{3}^{2}$, \dots, $w_{1}^{3}$, $w_{2}^{3}$,
    $w_{3}^{3}$, \dots
  \item[\rm{Predicate Symbols:}] a set of n--ary predicate symbols $r^n_i$ where $n \geq 2$ and $i \geq 1$: $r_{1}^{2}$,
    $r_{2}^{2}$, $r_{3}^{2}$, \dots, $r_{1}^{3}$, $r_{2}^{3}$, $r_{3}^{3}$, \dots, $r_{1}^{4}$, $r_{2}^{4}$, $r_{3}^{4}$, \dots
\item[\rm{Identity symbol:}] $\equiv$
\item[\rm{Tableau negation symbol:}] $\sim$.
  \end{description}


Some comment on  tableau negation is necessary. In most tableau systems there is no tableau negation 
but in some there is. 
Intuitively, the negation can be seen as a meta-level sign in a traditional signed tableau.
So, it is generally not needed. However to have, later, a universal notion of tableau inconsistency in the proof, we introduce $\sim$. When we examine a tableau system in which a branch is closed if it contains two expressions $A$ and $B$ of some kind, we always can think of an additional tableau rule that makes one more step adding $\sim A$ to the branch. So, any kind of tableau inconsistency can be reduced to appearance of  some pair: $A$ and $\sim A$ on a tableau level, but maybe not vice versa. Maybe sometimes we need the tableau negation to explicitly state an inconsistency of branch. Summing up, the systems with explicit negation can cover not less cases than the systems with no explicit negation.  

We also require some kind of terms to denote in the tableau language domains $\{ W_{i} \}_{i \in M}$ from semantic models. We call them \textit{labels} (in fact they are \emph{generalized labels}).

\begin{defi}\label{Tableu labels}[Tableau labels] The set of all labels $\TE$ is the least set that consists of  $w_{k}^{l}(m_{1}, \dots, m_{l})$, where:
$k, l \in \mathbb{N}, m_{1}, \dots, m_{l} \in \mathbb{I}$, $l \geq 1$, and  $w_{k}^{l}$ is a functional symbol. The members of  $\TE$  we denote by  $t_{1}, t_{2}, t_{3},$ \dots
\end{defi}

Next we define a set of tableau expressions whose members are just called \textit{expressions}.



\begin{defi}\label{Expression} [Expressions] For all formulae $A \in \For$, natural numbers $ k, l \in \mathbb{N}$,
    indices $i, j, m_{1}, \dots, m_{l}  \in \mathbb{I}$, $t \in \TE$, where $l \geq 2$, the set  $\textsf{Ex}$ is the least set that consists of: 

 \begin{itemize}

 \item $r_{k}^{l}(m_{1}, \dots, m_{l})$ \hspace{2cm} $\sim r_{k}^{l}(m_{1}, \dots, m_{l})$

 \item $i\equiv j$ \hspace{2cm} $\sim i\equiv j$

\item $\langle A, t \rangle$ \hspace{2cm} $\langle \sim A, t \rangle$.

 \end{itemize}

\end{defi}

\n  When the context is clear, we remove brackets $\langle$  $\rangle$ and write: $A, t$, or $\sim A, t$.

It seems a good moment to point out that our metatheory also examines non-labeled tableaux (just as it examines theories without a
direct semantic negation $\sim$, too). Notice that when we consider only expression $\langle A, t \rangle$ with the same label $t$, in fact we are dealing with non-labeled tableaux. Hence our further findings are also valid for this case.

Now, we would like to be able to  select indexes from sets of expressions.

\begin{defi}[Function choosing indexes] \textit{The function choosing indexes} is  the function $\circ \colon \textsf{Ex} \cup \Pow(\Ex) \longrightarrow \Pow(\mathbb{N})$ defined by condition, for all $A \in \For$, and $ k, l \in \mathbb{N}, i, j, m_{1}, \dots, m_{l} \in \mathbb{I}$, where $l \geq 1$:

\begin{itemize}

 \item $\circ(r_{k}^{l}(m_{1}, \dots, m_{l}))= \circ(\sim r_{k}^{l}(m_{1}, \dots, m_{l})) = \{m_{1}, \dots, m_{l} \}$

 \item $\circ(i\equiv j) = \circ(\sim i\equiv j) = \{ i, j \}$

\item $\circ(\langle A, w_{k}^{l} (m_{1}, \dots, m_{l})\rangle) = \circ(\langle \sim A, w_{k}^{l} (m_{1}, \dots, m_{l})\rangle) = \{m_{1}, \dots, m_{l} \}$

\item $\circ(X) = \bigcup \{ \circ(y): y \in X \}$, if $X \subseteq \Ex$,

\end{itemize}

\end{defi}


All tableau proofs can be rewritten with other indexes, however, preserving the presence  of $\For$ with varying indexes. But sometimes some subset of indexes $Z \subseteq \mathbb{I}$ may play a distinguish role (like logical values) and can not be replaced. So, we have a definition of similar sets of expressions with respect to some possibly empty  $Z$.

\begin{defi}[Similar sets of expressions]\label{Similar sets of expressions}

Let $X$, $Y \subseteq \Ex$ be sets of expressions. Let $Z \subseteq \mathbb{I}$. Set $X$ is  \textit{similar to} $Y$ \emph{with respect to}  $Z$ iff there is a~bijection $\ddag \colon \circ(X) \longrightarrow \circ (Y)$ such that  $\circ(X)$, $\circ(Y)$ are sets of indexes occurring in expressions of $X$ and $Y$ and 
for all $A \in \textsf{For}$ and $k, l \in \mathbb{N}, i, j, m_{1}, \dots, m_{l} \in \mathbb{I}$, where $l \geq 1$:

\begin{itemize}
  \item[(1)]  for all $x \in Z$, if $x \in \circ(X)$, then $\ddag(x) = x$

  \smallskip

  \item [(2)] for all kinds of expressions in \Ex:
\end{itemize}

\begin{itemize}

 \item[(a)] $r_{k}^{l}(m_{1}, \dots, m_{l}) \in X$ iff  $r_{k}^{l}(\ddag(m_{1}), \dots, \ddag(m_{l})) \in Y$

  \item[(b)] $i\equiv j \in X$ iff $\ddag(i) \equiv \ddag(j) \in Y$

  \item[(c)] $\langle A, w_{k}^{l} (m_{1}, \dots, m_{l}) \rangle \in X$ iff $\langle A, w_{k}^{l}(\ddag(m_{1}), \dots, \ddag(m_{l})) \rangle \in Y$

 \item[(d)]  $\sim r_{k}^{l}(m_{1}, \dots, m_{l}) \in X$ iff  $\sim r_{k}^{l}(\ddag(m_{1}), \dots, \ddag(m_{l})) \in Y$

 \item[(e)]   $\sim i\equiv j \in X$ iff $\sim \ddag(i) \equiv \ddag(j) \in Y$

  \item[(f)] $\langle \sim A, w_{k}^{l} (m_{1}, \dots, m_{l}) \rangle \in X$ iff $\langle \sim A, w_{k}^{l}(\ddag(m_{1}), \dots, \ddag(m_{l})) \rangle \in Y$,

\end{itemize}

\end{defi}


We need also a notion of a branch inconsistent set of expressions.

\begin{defi}[Branch inconsistent set of expressions]\label{Branch inconsistent set of expressions}

  Let $X \subseteq \textsf{Ex}$.   We say that $X$ is \textit{branch inconsistent} iff it contains at least two expressions which form a complementary pair of any of the following forms, for all $A \in \textsf{For}$ and $k, l \in \mathbb{N}, i, j, m_{1}, \dots, m_{l}, n  \in \mathbb{I}$, $t \in \TE$, where $l \geq 2$:

 \begin{itemize}


 \item[(a)] $r_{k}^{l}(m_{1}, \dots, m_{l})$, $\sim r_{k}^{l}(m_{1}, \dots, m_{l})$

 \item[(b)] $i\equiv j$, $\sim i\equiv j$

\item[(c)] $\langle A, t \rangle$, $\langle \sim A, t \rangle$,

 \end{itemize}

Otherwise, we call set $X$ \textit{branch consistent}. We abbreviate these to say that $X$  is \textit{b-inconsistent} or  respectively \textit{b-consistent}.

\end{defi}

So, a set can be b-consistent but it may expand to a b-inconsistent set. 
We just mean that  b-consistency is simply looking for a complementary pair. If it finds such a pair, then it says the set is   b-inconsistent. For example $\{\langle A, t \rangle, \langle \lnot\lnot\lnot A, t\rangle\}$ is b-consistent since there is no complementary pair. But if we later expand   this set via a double-negation cancelling rule say to $\{\langle A, t \rangle,$
$\langle \lnot\lnot\lnot A, t\rangle, \langle \lnot A, t\rangle, \langle \sim A, t \rangle \}$
then this set is b-inconsistent. So the first set is   b-consistent but it expands into a b-inconsistent set.

Let us observe that if two sets of expressions $X, Y \subseteq \Ex$ are similar, they are either both b-consistent, or both
b-inconsistent, by definitions~\ref{Similar sets of expressions} and \ref{Branch inconsistent set of expressions}.

\begin{fa}\label{Two similar sets of expression are either consistent or inconsistent} If two sets of expressions are similar, they are either both b-consistent, or both b-inconsistent
\end{fa}

Let $\mathfrak{M}= \langle \{W_{i}\}_{i\in M}, \{ R_{j}\}_{j\in N}, \vartheta \rangle \in \M$. We consider  functions $f \colon \mathbb{N}
\longrightarrow M \cup N$. Let us notice that:  if $f(i) = k$, for $k \in M$, and $W_{k}$ is a domain that consists of $n$-tuples in $\{W_{i}\}_{i\in M}$, then  $w_{i}^{n}$ is an n--ary functional symbol; if $f(i) = k$, for $k \in N$, and $R_{k}$ is an n--ary relation in $\{R_{j}\}_{j\in N}$, then  $r_{i}^{n}$ is an n--ary predicate symbol.

We assume some family $\{f_{\m}\}_{\m \in \M}$ of functions $f$ --- exactly one  for each model $\m \in \M$.
The function defined on the ground of $\{f_{\m}\}_{\m \in \M}$ by the condition $\f(\m, i) = f_{\m}(i)$ will be called \textit{the fitting function} and denoted by $\f$ (so $\f \colon \M \times \mathbb{N} \longrightarrow M \cup N$).
 When  a model $\m$ is fixed, then instead of $\f(\m, i)$ we write $\f(i)$.


We can now generalize the notion of interpretation  from the formulae to sets of expressions in a~model from $\M$.

\begin{defi}[Model suitable for a set of expressions]\label{Model sutable to a set of expressions}

Let $X \subseteq \Ex$ and $\m= \langle \{W_{i}\}_{i\in M}, \{ R_{j}\}_{j\in N}, \vartheta \rangle \in \M$. Let $\f \colon \M \times \mathbb{N} \longrightarrow M \cup N$ be the fitting function for $\M$. Model $\mathfrak{M}$ is \textit{suitable} for $X$ iff
 there is a function: $f': \mathbb{I} \longrightarrow \bigcup_{i \in M}  W_{i}$, such that the following conditions are fulfilled, for all  $A \in \For$ and $ k, l \in \mathbb{N}, i, j, m_{1}, \dots, m_{l} \in \mathbb{I}$, where $l \geq 1$:

\begin{itemize}



\item[(a)]  if $\langle A, w_{k}^{l}, (m_{1}, \dots, m_{l}) \rangle \in X$,  then $\langle f'(m_{1}), \dots, f'(m_{l}) \rangle \in W_{\f(k)}$ and \break $\m, \langle f'(m_{1}), \dots, f'(m_{l}) \rangle \models A$ 

\item[(b)]  if $\langle \sim A, w_{k}^{l}, (m_{1}, \dots, m_{l}) \rangle \in X$,  then $\langle f'(m_{1}), \dots, f'(m_{l}) \rangle \in W_{\f(k)}$ and \break $\m, \langle f'(m_{1}), \dots, f'(m_{l}) \rangle \not  \models A$ 

\item[(c)]  if $r_{k}^{l}(m_{1}, \dots, m_{l}) \in X$, then $\langle f'(m_{1}), \dots, f'(m_{l}) \rangle \in R_{\f(k)}$

\item[(d)] if $\sim r_{k}^{l}(m_{1}, \dots, m_{l}) \in X$, then $\langle f'(m_{1}), \dots, f'(m_{l}) \rangle \not \in R_{\f(k)}$


\item[(e)] if $i\equiv j \in X$, then $f'(i) = f'(j)$\footnote{Clearly $\equiv$ is a symbol of the tableau language, while = and $\neq$ are metalanguage symbols. In our framework they are, respectively: identity of denotation and non-identity of denotation of a pair of indexes under mapping $f'$.}

\item[(f)] if $\sim i\equiv j \in X$, then $f'(i)  \neq f'(j)$,

\end{itemize}

\end{defi}

It is clear that for a b-inconsistent set of expressions no model is suitable. So, by \ref{Branch inconsistent set of expressions},  \ref{Model sutable to a set of expressions} we have.

\begin{fa}\label{Fact No model is suitable for b-inconsistent set} No model is suitable for a b-inconsistent set of expressions.
\end{fa}

\section{Rules, branches and tableaux}\label{sec:4}

The crucial notion of any tableaux system is a notion of tableau rule. Here we propose a very abstract view on this notion. Having a set of expressions $\textsf{Ex}$ we can establish some general conditions for defining rules. We start with the initial notion of rule. Assume that $\Pow (\textsf{Ex})$  is the set of  all subsets of the set $\textsf{Ex}$. Let
$\Pow(\textsf{Ex})^{n}$ be an
$n$--ary Cartesian product $\underbrace{\Pow(\textsf{Ex})\times \dots \times \Pow(\textsf{Ex})}_{n}$, for some
$n \in \mathbb{N}$, and let $\bigcup_{n \in \mathbb{N}} \Pow(\textsf{Ex})^{n}$  be the union of all such $n$--ary Cartesian
products such that $n \geq 2$.

\begin{defi}[Rule]\label{Rule}

 \textit{Rule} is a subset $\R \subseteq \bigcup_{n \in \mathbb{N}} \Pow(\textsf{Ex})^{n}$ such that
if $\langle X_{1}, \dots, X_{n} \rangle$ $\in \R$, then the following conditions are satisfied:

\begin{itemize}

\item[(a)] $X_{1} \subset X_{i}$, for all $1 < i \leq n$

\item[(b)] $X_{1}$ is $b$-consistent

\item[(c)] for all $1 < k \leq n$ and $1 < l \leq n$,  if $k \neq l$, then $X_{k} \neq X_{l}$.

\end{itemize}

\n We often refer to the set $X_{1}$ as the \emph{input} and to the sets $X_2, X_3, \cdots , X_n$ as the \emph{outputs} of a given rule. \end{defi}

Our rules extend properly a set of expressions by (a): an input is always a proper subset of outputs, so a rule does not work
trivially. By (b), any input is always b-consistent, so no rule can be used senselessly, extending a set that cannot generate a
counter-model. (So, rules have an internal mechanism that blocks extending b-inconsistent sets.)  Finally, by (c), rules do not
contain duplicated outputs.

Let us distinguish a set of indexes $\LV \subseteq \mathbb{I}$ that plays the role of signs of logical values in our language (they denote objects of some domain in a model: see example \ref{Example mulitivalued and intensional}) or other roles that require constants. 
The set $\LV$ can be empty in case we have a two valued semantics or we do not want to incorporate values in the labels.  One more helpful notion to finish. Let $Z \subseteq \Ex$. Then $Z$ is \emph{co--infinite}  iff $\mathbb{N} \setminus \circ(Z)$ is infinite.

\begin{defi}[Tableau rule]\label{Tableau rule}

Let $\R$ be a rule and  $\langle X_{1}, \dots, X_{n} \rangle$ $\in \R$. $\R$ is \emph{a tableau rule} iff:

\begin{itemize}

\item[(CS)] $\R$ is \emph{closed under similar sets}: for any $Y_{1} \subseteq \Ex$, if $Y_{1}$ is similar
to $X_{1}$ with respect to  $\LV$, then there exist sets of expressions  $Y_{2}$, \dots, $Y_{n}$, such that $\langle Y_{1}$, $\dots$, $Y_{n} \rangle \in \R$  and for all $1 < i \leq n$, $Y_{i}$
is similar to $X_{i}$ with respect to $\LV$

\item[(CF)] $\R$ is \emph{closed under finite sets}: if $X_{1}$ is a finite set, then for all $1 < i \leq n$, $X_{i}$ is a finite set

\item[(CC)] $\R$ is \emph{closed under cores}: for some finite set $Y \subseteq X_{1}$ there exists exactly one  $n$--tuple $\langle Z_{1}$, $\dots$, $Z_{n} \rangle \in \R$ such that:

\begin{enumerate}

\item  $Z_{1} = Y$

\item for all  $1< i \leq n$, $Z_{i} = Z_{1} \cup (X_{i} \setminus X_{1})$

\item there does not exist a proper subset $U_{1} \subset Y$ and
$n$--tuple $\langle U_{1}$, $\dots$, $U_{n} \rangle \in \R$ such that for all $ 1 < i \leq n$, $U_{i} = U_{1} \cup (Z_{i} \setminus Z_{1})$

\end{enumerate}

Any such $n$--tuple $\langle Z_{1}$, $\dots$, $Z_{n} \rangle$ is called \emph{a core of rule} $\R$ \emph{in}  $\langle X_{1}, \dots, X_{n} \rangle$.
\end{itemize}

\begin{itemize}

\item[(CE)] $\R$ is \emph{closed under expansion}:  for any $b$-consistent set of expressions $Z_{1}$ such that:

\begin{enumerate}

\item  $X_{1} \subset Z_{1}$

\item $Z_{1}$ is co--infinite

\item for all $1 < i \leq n$, $X_{i}$ is not similar with respect to $\LV$ to any subset of $Z_{1}$,

\end{enumerate}

 if $n$--tuple $\langle W_{1}$, $\dots$, $W_{n} \rangle$ is a core of rule $\R$ in $\langle X_{1}, \dots, X_{n} \rangle$, then:

\begin{enumerate}

\item  there are  $n - 1$  sets of expressions $Z_{2}$, \dots, $Z_{n}$ such that $\langle Z_{1}, \dots, Z_{n} \rangle \in \R$

\item and for all $1 < i \leq n$, $W_{i}$ is  similar with respect of $\LV$ to $W_{1} \cup (Z_{i} \setminus Z_{1})$.

\end{enumerate}

\end{itemize}

\end{defi}

By saying that a rule $\R$ \textit{was applied to} $X_{1}$ \textit{to obtain} $Y$, we mean that there exists one $\langle X_{1},
\dots, X_{n} \rangle \in \R$ such that for some $1 < i \leq n$,  $Y = X_{i}$.

The explanation of both definitions is included in the following example of rules written in the proposed manner.

\begin{ex}\begin{myfont}[Tableau rules]\label{Example: tableau rules}
  First, we comment on definition~\ref{Rule}. To make it easy to add comments we assume that a rule is represented as a schema-fraction where $2 \leq n$: 
\begin{align*}
\begin{array}{c}
  X_1 \\ \hline
  X_{2} | \dots | \ X_{n}
\end{array}
\end{align*}

According to definition~\ref{Rule}, instantiations of $X_{1}, \dots, X_{n}$ are subsets  of $\Ex$.
As we see, rules are n-tuples of sets of expressions, in fact they are intended to extend sets, not uniquely, if we have more outputs. However, only in ways limited by conditions (a), (b), (c) of Definition~\ref{Rule}.

Not all rules are tableau rules. Definition \ref{Tableau rule} specifies four properties that make a tableau rule from a rule. Let
us inspect them via the examples.  We can start with some possible rule $(\R_{\vee})$ for propositional connective $\vee$:

\begin{align*}
 \begin{array}{c}
  X \cup \{ \langle A \vee B, w(i, j) \rangle \} \\ \hline
  X \cup \{\langle A \vee B, w(i, j) \rangle, \langle A , w(i, j) \rangle \} ~|~ X \cup \{\langle A \vee B, w(i, j) \rangle, \langle B , w(i, j) \rangle \}
\end{array}
\end{align*}

\noindent where $A \vee B$ is a disjunction of two formulae, $i \in \mathbb{I}$, $j$ is a fixed and designated logical value from $\LV$, $\langle C, w(i, j) \rangle$ codes that formula $C$ has in world $i$ the value $j$, and  set $X \cup \{ \langle A \vee B, w(i, j) \rangle \}$ is b-consistent.

It looks complicated, but it is a formal and paradigmatic form of the rule. If we assume that $(\R_{\vee})$ fulfills the
properties of definition~\ref{Rule}, than we can slim it down. When we omit also brackets $\langle, \rangle$, it can be written
as:

\begin{align*}
(\text{\R}_{\vee})\quad\quad \begin{array}{c}
 A \vee B, w(i, j) \\ \hline
   A , w(i, j) ~|~  B , w(i, j)
\end{array}
\end{align*}

\noindent Nonetheless, we still can apply it only to sets that are b-consistent and do not include the output. So far, we have
exemplified definition~\ref{Rule}.

Obviously, if we know that $i$ is any point of relativization, and $j$ is a~fixed, designated  value, we can also abandon functional symbol $w$ (it is only for a general, metatheoretical use to denote  the correct domain in the structure):

\begin{align*}
(\text{\R}_{\vee})\quad\quad \begin{array}{c}
 A \vee B, i, j  \\ \hline
   A , i, j ~|~  B , i, j
\end{array}
\end{align*}

\noindent  Going further, if we work in a two-valued logic and interpret $\vee$ as the classical disjunction, we can also omit $j$, and when we have only one world, one point of relativization (so a non-modal logic), we can do the same with an arbitrary $i$, assuming just:

\begin{align*}
(\text{\R}'_{\vee})\quad\quad \begin{array}{c}
 A \vee B \\ \hline
   A  ~|~  B
\end{array}
\end{align*}

Up to now, by example of $(\R_{\vee})$, and its particularization  $(\R'_{\vee})$, we have illustrated three of the conditions for a rule being a tableau rule according to definition~\ref{Tableau rule}. The rule is closed under similar sets (CS), because $(\R_{\vee})$ is invariant with respect to the designated logical value denoted  by $j$, but defined for other indexes not already reserved to denote any logical values from $\LV$. The same applies to being closed under finite sets (CF). By  rule $(\R_{\vee})$, we always extend input to outputs by adding a finite set of expressions. So, any step of the proof by rule satisfying  (CF) must introduce a finite piece of information to the proof. Finally, the last property (CC). All  cores  of rule $(\R_{\vee})$ are  ordered triples:  $ \langle \{ \langle A \vee B, w(i, j) \rangle \},  \{ \langle A , w(i, j) \rangle \}, \{ \langle B , w(i, j) \rangle \} \rangle$. Although, we apply the rule to the first member, we must be careful about the whole set of expressions. The existence of cores enables us to maintain a balance between single additions, atoms of our reasoning and the~whole set of information on a branch. Thanks to this we are sure  that the use of rule is non-trivial, since additional clauses (b-consistency, being superset, etc.) refer to the whole branch.

An additional clause is also introduced by being closed under expansion (CE). We shall explain it by analysing another rule whose focus is the intensional case  $\Diamond A$. We would  like to understand $\Diamond A$  as a sentence with the possibility operator, but in a many-valued frame, under some specific, exemplary meaning. 


\begin{align*}
(\text{\R}_{ \Diamond})\quad\quad \begin{array}{c}
 X \cup \{\langle  \Diamond A,  w(i, j) \rangle \}\\ \hline
 X \cup \{ \langle  \Diamond A,  w(i, j) \rangle, r(i,k), \langle A, w(k, j) \rangle\}
\end{array}
\end{align*}


\noindent We assume that $j$ is a designated logical value and fixed. The respective expressions mean:

\begin{itemize}

\item $\langle  \Diamond A,  w(i, j) \rangle$ ---  sentence $\Diamond A$ in world $i$ possesses the logical value $j$,

\item $r(i,k)$ ---  world $k$ is accessible from world $i$,


 \item $\langle A, w(k, j) \rangle$ --- sentence $A$ in world $k$ possesses logical value $j$.

\end{itemize}

\noindent The meaning seems clear. But the rule requires two clauses:

\medskip

(1) $k \not \in \circ(X \cup \{\langle  \Diamond A,  w(i, j) \rangle \})$

\smallskip

(2)  for all $k' \in \mathbb{N}$, $\{r(i, k'), \langle A, w(k', j) \rangle\} \not \subseteq X$.

\medskip

\noindent Clause (1) obliges to introduce  a fresh label in case of use of  $(\text{\R}_{ \Diamond})$ to denote a world. It  must be fresh in a proof   to ensure that we cannot obtain a   b-inconsistency simply by reusing the existing values for these indexes. Clause (2) is to prevent us from an~unnecessary use of $(\text{\R}_{ \Diamond})$, if  we already have sufficient conditions to verify $\Diamond A $ at world $i$ with value
$j$. In fact none of the clauses can be formally imposed, if the form of rule $(\text{\R}_{ \Diamond})$ is:

\begin{align*}
(\text{\R}'_{ \Diamond})\quad\quad \begin{array}{c}
 \Diamond A, i, j \\ \hline
 r(i,k) \\
 A, k, j 
\end{array}
\end{align*}

\noindent since we have no reference to premisses other than merely $\Diamond A, i, j$ (a branch or a tableau are beings of a higher order, dependent on the definition of rules, so formally we can not impose any clauses on them!).
The clauses of rule $(\text{\R}_{ \Diamond})$ justify (CE) well. The condition of being closed under expansion in
definition \ref{Tableau rule} says that if a rule is defined for some input $X_{1}$, then under several restrictions it is also
defined for any superset $Z_{1} \supseteq X_1$. Of course, some restrictions (like being b-consistent) are standard in the light of what we have already said. But two of them are essentially new.

First $Z_{1}$ must be co-infinite.  That is, there are still an infinite, unused number of indexes $\mathbb{I} \setminus \circ(Z_{1})$.  So, if we need to introduce new ones to outputs, we can. It is clear that the restriction corresponds to clause (1). The existence of expansion of a given rule should not be expected, if the specific clause requiring new indexes can not be fulfilled, because no fresh indexes are available.

Second, for all  $1 < i \leq n$, $X_{i}$ (output of input $X_{1}$) is not similar to any subset of $Z_{1}$ (with respect to $\LV$,  which is invariant). Otherwise, $Z_{1}$  would provide the same semantic piece of information as  some $X_{i}$, and it would not make sense to use the rule on $Z_{1}$ (in fact, on the core that $X_{1}$ and $Z_{1}$ share). Clearly, this corresponds to clauses like clause (2), where we have enough information to not use the rule. So, once again, there is no need for expansion here.

It is worth noting that the conditions (CS), (CF), (CC), (CE), which should be met by a tableau rule according to the definition \ref{Tableau rule} of tableau rule are usually met by all rules defined in tableau systems. For example, the condition (CS) is responsible for the structurality of the rule, the condition (CF) for the fact that the effect of applying the rule is a finite set, the condition (CC) for the existence of a minimal set of expressions to which we apply the rule, while the condition (CE) says that the rule can be applied to supersets of minimal sets of expressions, provided that this brings new information. 

However, these conditions have never been formulated generally and abstractly, which allows us in section \ref{sec:5} to demonstrate some metatheoretical dependencies in general, assuming only that the tableau rules satisfy conditions (CS), (CF), (CC), (CE) and are rules according to the definition of a rule \ref{Rule}.
\end{myfont}
\qed{}
\end{ex}

The restrictions proposed in definitions of rule \ref{Rule} and tableau rule \ref{Tableau rule} which were mentioned in example \ref{Example: tableau rules} are not only in place to render the proofs more economical. First of all, the reason is that at the metatheoretical level, the proofs of general facts about our systems rely on certain structural properties of rules. The properties we have introduced are sufficient (but probably not necessary) for these purposes (see section \ref{sec:5}).

A further assumption we make is that by $\TR$ we denote some fixed set of tableau rules. Having a fixed $\TR$ we can work on more complex tableau notions, such as those of a branch and a tableau.

In the traditional approach a tableau  is usually defined first and then branches are extracted from it. Our proposal here is reverse.  We propose a multistage, set-theoretical  construction of tableau notions. We will deal with three levels of ever more complex set-theoretical objects in our theory: (1) tableau rules (sets of n-tuples), (2) branches (special sequences of sets), (3)   tableaux (special sets of branches). The reader may ask how suitable set of branches can be included in the same tableau.
The answer is the condition of cohesion given in Definition~\ref{Tableau} of tableau.
Furthermore, the notion of complete tableau from Definition~\ref{Complete incomplete tableau} answers the question. It is worth saying that the view we present here is translatable into the traditional view (by graphs or pictures).

Intuitively, a  branch is a sequence of sets of expressions: $X_{1} \subseteq X_{2} \subseteq \dots \subseteq X_{n}$
(possibly infinite), where for any $1 \leq i < n$, $X_{i+1}$ is a result of application of some tableau rule to set $X_{i}$. Formally:

\begin{defi}[Branch]\label{Branch}  Let $K = \mathbb{N}$ or $K = \{1, 2, \dots, n \}$, where $n \in \mathbb{N}$.
Let $X$ be a set of expressions. \textit{A branch} (or \textit{a branch starting from} $X$) is any sequence
 $\phi: K \longrightarrow \Pow (\Ex)$ satisfying the conditions:

\begin{enumerate}

\item $\phi(1) = X$

\item for all $i \in K$: if $i+1 \in K$, then there exists a rule $\R \in \TR$   and $n$-tuple
$\langle Y_{1}, \dots Y_{n}\rangle$ $\in$ $\R$, such that $\phi(i) = Y_{1}$ and $\phi(i+1) = Y_{k}$, for some $1 < k \leq n$.

\end{enumerate}

\end{defi}

We assume that branches of a given logic will be denoted by small Greek letters $\phi, \psi, \chi$ etc., while sets of branches by big Greek letters $\Phi, \Psi$  etc.

\noindent   From the definition  of a branch \ref{Branch} we have a corollary:

\begin{coro}\label{Coro Set of expressions is branch}

$\bullet$ Let $X \subseteq \Ex$. Then $\phi \colon \{ 1 \} \longrightarrow \Pow(\Ex)$, where $\phi(1) = X$, is a branch.

$\bullet$ Let $\phi\colon  K \longrightarrow \Pow(\Ex)$ be a branch.  Then $\psi \colon \{ 1 \} \longrightarrow \Pow(\Ex)$, where $\psi(1) = \bigcup \{ \phi(i) | i \in K \}$,  is also a~branch. 

\end{coro}

\noindent Taking into account corollary \ref{Coro Set of expressions is branch}, where we make a union of a branch, instead of $\psi(1) = \bigcup \{ \phi(i) | i \in K \}$ we will  write $\bigcup \phi$ as a one--member--branch made from  the union  of the members of the branch $\phi:  K \longrightarrow \Pow(\textsf{Ex})$.

Different kinds of branches we can now define below.

 \begin{defi}[Closed/open/complete branch]\label{Branch: closed/open,complete} Let $\phi: K \longrightarrow \Pow(\textsf{Ex})$ be a branch.

  \begin{itemize}

\item $\phi$ is \textit{closed} iff there exists $i \in K$ such that $\phi(i)$
  is a b--inconsistent set;   $\phi$ is \textit{open} iff it is not closed

\item $\phi$ is \textit{complete} iff  there is no branch $\psi:  \{1, 2\} \longrightarrow \Pow (\textsf{Ex})$ such that $\bigcup \phi = \psi(1)$.
\end{itemize}

\end{defi}

\noindent Note that if $\phi(i)$ is a b--inconsistent set, for some $i \in K$, then it is the last  element of a branch  $\phi$, because rules cannot be applied to b-inconsistent sets,  by \ref{Rule}. So, a closed branch is always finite. The definition of complete branch says in turn that a complete branch $\phi$ has a full information in respect of the tableau rules $\TR$. If we take the union $\bigcup \phi$ of the whole chain there is no possibility to extend it by means of any tableau rule from $\TR$. If the reader asks why we define a complete branch in this way, we will answer that branch $\phi$ can be infinite, but not complete. If after summing $\bigcup \phi$ no tableau rule applies, then $\phi$ is complete. If the opposite is true, then although $\phi$ is infinite, it is not complete. This definition therefore also covers the case of infinite branches.

Let $t \in \TE$ and $X \subseteq \For$. By $X^{t}$ we denote the set $\{\langle A, t \rangle : A \in X \}$.  \label{Remarks on starting label} Additionally, we assume  a label  $\t$ as a starting point of branches and tableaux. Let us recall that the label $\t$ must have the form $w_{j}^{i}(x_{1}, \dots, x_{i})$,  for some functional symbol  $w_{j}^{i}$ and $x_{1}, \dots, x_{i} \in \mathbb{I}$.  It is also assumed that $w_{j}^{i}$ under the fitting function $\f$ is related to the k$^{th}$ domain $W_{k}$ in models from $\M$. We can now introduce the notion of branch consequence relation which is a~tableau counterpart of $\models_{\M}$.

\begin{defi}[Branch consequence relation]\label{Branch consequence relation}

Let $X \cup \{ A \}\subseteq \textsf{For}$. Formula $A$ is a \textit{branch consequence of} $X$
(in short: $X \vartriangleright_{\TR} A$) iff there exists  a~finite set $Y \subseteq X$, such that all complete branches starting
from the set $Y^{\t}\cup\{\langle \sim A, \t \rangle\}$ are closed.

\end{defi}

It would seem that there are infinitely many complete branches that we would need to check to know if $X \vartriangleright_{\TR} A$. In practice, however, a tableau is just a minimal and  sufficient set of branches that we need to check, so we could treat the construction of a complete tableau as a way of choosing a sufficient number of branches from a much wider universe of complete branches.

We say a branch $\psi$ is \textit{maximal in  set of branches} $\Phi$ (in short: $\Phi$--\textit{maximal}) iff there does not exist  a
branch $\phi \in \Phi$, such that $\psi \subset \phi$.

\begin{defi}[Tableau]\label{Tableau}
Let $X \cup \{A\}\subseteq \textsf{For}$ and $\Phi$ be a set of branches.
An ordered triple $\langle X, A, \Phi \rangle$ is \textit{a tableau} for $\langle X, A \rangle$
(or just \textit{tableau}) iff:

\begin{itemize}

\item $\Phi$ is a non-empty subset of the set of branches starting from the set
 $ X^{\t} \cup \{\langle \sim A, \t \rangle \}$,
 (i.e. if $\psi \in \Phi$, then $\psi(1) = X^{\t} \cup \{\langle \sim A, \t  \rangle \}$)

\item each branch in $\Phi$ is $\Phi$-maximal

\item for any $n, i \in \mathbb{N}$ and any branches
$\psi_{1}$, \dots, $\psi_{n} \in \Phi$, if:

\begin{itemize}

\item $i$ and $i+1$ are in domains of functions $\psi_{1}$, \dots, $\psi_{n}$
\item for any $1 \leq  k, k' \leq n$ and any $o \leq i$, $\psi_{k}(o) = \psi_{k'}(o)$

\end{itemize}

\noindent then there exists  a rule $\R \in \TR$ and an ordered $m$-tuple
$\langle Y_{1}$, $\dots$, $Y_{m} \,\rangle$ $\in$ $\R$, where $1 < m$,  such that for all $1 \leq k \leq n$:

\begin{itemize}

\item $\psi_{k}(i) = Y_{1}$

\item there exists  $1 < l \leq m$, such that $\psi_{k}(i+1) = Y_{l}$.

\end{itemize}

\end{itemize}

\end{defi}

According to definition \ref{Tableau}, a tableau is a set of branches vertically and horizontally ordered by the set of rules $\TR$ and updated by subsequent applications. Since
we would like to define a~complete tableau as a maximal set of branches restricted by certain conditions,  we must take into account a~problem concerning branches that can (but may not) occur in a complete tableau. Let us consider an example.

\begin{ex}\begin{myfont}[Redundant branches]\label{Example: redundant branches} Consider a tableau system for a logic with two logical values (truth and falsity) and with only one point of relativization. 
Let classical negation $\neg$ and classical disjunction $\vee$  belong to the set  $\Con$ of symbols of the logic. Finally, let the set of tableau rules among others consist of symbolically written rules:

\begin{align*}
(\text{\R}_{\neg \neg})\quad\begin{array}{c}
 \neg \neg A \\ \hline
   A
\end{array}\qquad\qquad\qquad(\text{\R}_{\vee})\quad \begin{array}{c}
 A \vee B \\ \hline
   A  \mid B
\end{array}
\end{align*}

\n Now, assume some set $X \cup \{ p \vee q,  \neg \neg p \}$, where $X$ is a set of formulae and $p, q$ are variables. Look at these three sequences that are branches (due to the definition of branch \ref{Branch}, if $(\text{\R}_{\neg \neg})$ and $(\text{\R}_{\vee}) \in \TR$):

\begin{description}
\item[\rm (a):] $X \cup \{ p \vee q,  \neg \neg p \} \subset X \cup \{  p \vee
q, \neg \neg p,  p  \}$, produced by the application of rule $\R_{\neg\,\neg}$ 
\item[\rm (b):] $X \cup \{ p \vee
q, \neg \neg p \} \subset X \cup \{p \vee q,  \neg \neg p, p \}$
produced by the application of the rule for disjunction $\R_{\vee}$  (to the left component of disjunction) 
\item[\rm (c):] $X \cup \{ p \vee q, \neg \neg p \} \subset X \cup \{ p \vee q, \neg \neg p, q
\}  \subset X \cup \{ p \vee q, \neg \neg p, q, p\}$, produced first by the application of the rule for disjunction $\R_{\vee}$  (to the right component of disjunction), and then additionally extended by rule $\R_{\neg\,\neg}$.
\end{description}

\n Branch (c) is redundant in the following sense: if branch (a) is closed then (c) is closed; if (c) is open, then branch (a) is open, too. Hence all we need to know is contained in branch (a) (or (b) --- they are identical). As we see the presence of branch (c) in a tableau is acceptable, but not obligatory to get the full information on this part of proof.\end{myfont}
\end{ex}

We can now capture the phenomenon of redundant branches generally.

\begin{defi} [Redundant variant of branch]\label{Useless variant of branch} Let $\phi$ and $\psi$ be  branches such that for some numbers
$i$ and $i +1$ that belong to their domains  and for all $j \leq i$, $\phi(j) = \psi(j)$,
but $\phi(i+1) \not = \psi(i+1)$. Branch $\psi$ is \textit{a redundant variant of branch} $\phi$ iff:

\begin{itemize}

\item there are   rule $\R \in \TR$ and an
  $n$--tuple $\langle X_{1}, \dots, X_{n} \rangle \in \R$, such that $ \phi(i) = X_{1}$ and $\phi(i+1) = X_{j}$,
for some $1 < j \leq n$

\item there are rule $\R' \in \TR$ and an $m$--tuple $\langle Y_{1}, \dots, Y_{m} \rangle \in \R'$,
where $n < m$, such that $\psi(i) = Y_{1}$ and:

     \begin{enumerate}

     \item $\psi(i+1) = Y_{k}$, for some $1 < k \leq m$

     \item for all $1< l \leq n$ there is  $1 < o \leq m$ such that $o \not = k$ and~$X_{l} = Y_{o}$.

     \end{enumerate}

\end{itemize}

\noindent Let $\Phi$, $\Psi$ be sets of branches and $\Phi \subset \Psi$. $\Psi$ is \textit{a redundant superset} of $\Phi$ iff for any branch $\psi \in \Psi \setminus \Phi$ there is  a branch $\phi \in \Phi$ such that $\psi$ is a redundant variant of $\phi$.

\end{defi}

\noindent By use of \ref{Useless variant of branch} a  complete tableau can be defined by some kind of maximization.

\begin{defi}[Complete/incomplete tableau] \label{Complete incomplete tableau} \!\!Let \!$\langle X, A, \Phi \rangle$ be a tableau. A tableau $\langle X, A, \Phi \rangle$ is \textit{complete} iff:

\begin{enumerate}

\item all branches in $\Phi$ are complete

\item if  $\Phi \subset \Psi$ and $\langle X, A, \Psi \rangle$ is a tableau, then $\Psi$ is a redundant superset of $\Phi$, for all sets of branches $\Psi$.

\end{enumerate}

\noindent A tableau is \textit{incomplete} iff it is not complete.

\end{defi}

\begin{defi}[Closed/open tableau]\label{Closed open tableau}

Let $\langle X, A, \Phi \rangle$ be a tableau. A tableau $\langle X, A, \Phi \rangle$ is \textit{closed} iff $\langle X, A, \Phi \rangle$ is complete and all branches in $\Phi$ are closed. A tableau is \textit{open} iff it is not closed.

\end{defi}

\begin{ex}\begin{myfont}[Tableaux for a variant  of $\SEp$]\label{Przyklad tableaux dla podlogiki S} Let us consider some attempt to construction of tableau system for the logic presented in example \ref{Przyklad podlogiki S}. We do not claim here that the tableau system is sound and complete to that logic. It would require an additional proof.

To define tableaux rules and proofs for the logic determined by the models $\langle \vartheta, s \rangle$, where $\vartheta \colon \Var \longrightarrow \{0, 1 \}$, and $s \colon \For \longrightarrow W$, where $W$ is an arbitrary non-empty set of sets, and for a complex formula $A$:

\[ s(A) =  \bigcap \{U \in W | \text{ for some } p \in \var(A),  s(p) = U \}, \]

\n under the truth conditions given in \ref{Przyklad tableaux dla podlogiki S}, we remind that the set of formulae $\For$ is built with variables $\Var$ and connectives: $\neg$, $\wedge$, $\rightarrow$. Since it is a logic with two logical values (truth and falsity) and with only one point of relativization,  the set of tableau  expressions can be reduced to the language  of $\For$ extended by auxiliary expressions which  is the union of the following sets:
\begin{itemize}\itemsep=2pt

\item $\For\times\N$
    
\item $\{\langle \sim A, i \rangle \colon \langle A, i \rangle \in \For\times\N \}$.

\end{itemize}

\n The expressions $\langle A, i \rangle$ state that an object $i$ is in the content of $A$, while the expressions $\langle \sim A, i \rangle$ state that an object $i$ is not in the content of formula $A$. As we said earlier, when the context is clear, we remove brackets $\langle A, i \rangle$, $\langle \sim A, i \rangle$, and write: $A, i$, $\sim A, i$.

The notation for the rules includes the default assumption of all clauses from definitions of rule \ref{Rule} and tableau rule \ref{Tableau rule} (so, \textit{inter alia},  outputs are proper supersets of inputs and inputs are always b-consistent). We have classical rules for $\vee$, $\neg$ and their combinations:

\begin{align*}
(\text{\R}_{\wedge})\quad \begin{array}{c}
 A \wedge B \\ \hline
   A  \\ B
\end{array}\qquad\qquad\qquad(\text{\R}_{\neg \neg})\quad\begin{array}{c}
 \neg \neg A \\ \hline
   A
\end{array}\qquad\qquad\qquad(\text{\R}_{\neg \wedge})\quad \begin{array}{c}
 \neg(A \wedge B) \\ \hline
   \neg A  \mid \neg B
\end{array}
\end{align*}

\n For the implication we assume two rules. $(\text{\R}_{\rightarrow (1)})$ reflects the truth-functional aspect of $\rightarrow$, while $(\text{\R}_{\rightarrow (2)})$ expresses its content-related aspect. In case of $(\text{\R}_{\rightarrow (2)})$ we add additional clauses: 

\begin{flalign*}
&(\text{\R}_{\rightarrow (1)})& &\quad \begin{array}{c}
 A \rightarrow B \\ \hline
   \neg A  \mid B
\end{array}&\qquad\qquad\qquad(\text{\R}_{\rightarrow (2)})& & &\quad\begin{array}{c}
  A \rightarrow B \\ \hline
  A, i \\
  B, i
\end{array}\\
& & & \;\, \quad
& & & &\text{\small where: }&\\[-0.1cm] 
& & & \;\, \quad
& & & &\text{\small (1) $i$ is new}&\\[-0.1cm] 
& & & & & & &\text{\small (2) no expressions of the form $A, j$, }&\\[-0.1cm]
& & & & & & & \text{\small $B, j$ belong together to the input}&
\end{flalign*}



\n For  the negative case of implication we have:

\begin{flalign*}
& (\text{\R}_{\neg \rightarrow(1)})  & & 
\begin{array}{c}\neg(A\rightarrow B)\\ \hline \begin{array}{c|c} A & A, i\\ \neg B  & B, j \end{array}\end{array}&\;\, \quad\quad
&(\text{\R}_{\neg \rightarrow (2)}) & & 
\begin{array}{c}\begin{array}{c} A, i\\ B, j\\ A, k \end{array} \\ \hline \sim\! B, k \end{array}&\\[0.3cm]
& & &\text{\small where  $i,j$ are new}&\\[-0.1cm]
& & &\text{\small and $i\neq j$}&\\[-0.1cm]
& & &\text{\small and for no  $k \neq l$, the expressions: }\;\, \quad
& & & &\text{\small where $i,j$ are introduced by }&\\[-0.1cm] 
& & &\text{\small $A, k$ and $B, l$  belong to the input} & & & &\text{\small by the rule  (\text{\R}$_{\neg \rightarrow(1)}$)}&
\end{flalign*}

\n We have the rules that reflect the property of  \textit{intersected  set assignment}:

\begin{flalign*}
&\hspace{0.25cm} (\text{\R}_{i}) & &
\begin{array}{c} A, i \\ \hline\\[-1.2em]\begin{array}{c}   A_{1}, i \\ \vdots \\  A_{n}, i  \end{array}\end{array}\;\, \quad
& (\text{\R}_{\sim i}) & & &
\begin{array}{c} \sim A, i \\ \hline\\[-1.2em]\begin{array}{l|c|r}   \sim A_{1}, i & \dots & \sim  A_{n}, i  \end{array}\end{array}  &\\[0.2cm]
& & & \text{\small where $\var(A) = \{A_{1}, \dots, A_{n} \}$}
& & & & \text{\small where $\var(A) = \{A_{1}, \dots, A_{n} \}$}&
\end{flalign*}

\n And finally we must assume two rules to be close to the general paradigm we are establishing. One that transforms the metalinguistic negation $\sim$ into the classical negation we have in the language, and additionally a rule that reduces the classical contradiction to b-inconsistency. 

\begin{flalign*}
&\hspace{0.25cm} (\text{\R}_{\sim}) & &
\begin{array}{c} \sim A \\ \hline\\[-1.2em]  \neg A \end{array}\;\, \quad
& (\text{\R}_{\neg}) & & &
\begin{array}{c} A \\ \neg A  \\ \hline\\[-1.2em]  \sim A \end{array}  &\\[0.2cm]
\end{flalign*}

\n It is worth underlying that the clauses of the rules: $(\text{\R}_{\rightarrow (2)})$,  $(\text{\R}_{\neg \rightarrow (1)})$, $(\text{\R}_{\neg \rightarrow (2)})$ refer to the whole inputs, as rules are generally defined on sets of expressions.

We now  consider the set of expressions $\{p, \sim ( p  \rightarrow \neg (\neg p \wedge \neg q)) \}$.  
We have 33 kinds of branches obtainable by the use of some orders of applications  of the proposed tableau rules (according to the definition of branch \ref{Branch}). We write `33 kinds', because we will use a variable $i$, where $i  \in \mathbb{I}$. To not write too long sequences we assume some convention. If $\phi = X_{1} \subset \dots \subset X_{n}$ 
is a branch, by $\phi, Y$ we will mean the branch extended with $Y$: $X_{1} \subset \dots \subset X_{n} \subset Y$.

So, in fact, countably-infinitely many branches of these  33 kinds of branches exist, that  vary depending on what number variable $i$ takes. Let us start with the initial branches:

\medskip

\n $\phi_{0}$: $\{p, \sim ( p  \rightarrow \neg (\neg p \wedge \neg q)) \}$ --- a one-element branch, obtained by no use of the rules;

\medskip

\n $\phi_{1}$: $\{p, \sim ( p  \rightarrow \neg (\neg p \wedge \neg q)) \} \subset \{p, \sim ( p  \rightarrow \neg (\neg p \wedge \neg q)), \neg ( p  \rightarrow \neg (\neg p \wedge \neg q))\}$ --- a branch with two  elements,  obtained by the application of $(\text{\R}_{\sim})$ to $\sim ( p  \rightarrow \neg (\neg p \wedge \neg q))$ --- it is the only  way of  extending $\phi_{0}$;

\medskip

\n Now, applying the rule $(\text{\R}_{\neg \rightarrow(1)})$ to  $\neg ( p  \rightarrow \neg (\neg p \wedge \neg q))$ in the branch $\phi_{1}$, we obtain two branches:

\medskip

\n $\phi_{1.1}$: $\phi_{1}$, $\{p, \sim ( p  \rightarrow \neg (\neg p \wedge \neg q)), \neg ( p  \rightarrow \neg (\neg p \wedge \neg q)), p, \neg \neg (\neg p \wedge \neg q)\}$

\medskip

\n $\phi_{1.2}$: $\phi_{1}$, $\{p, \sim ( p  \rightarrow \neg (\neg p \wedge \neg q)), \neg ( p  \rightarrow \neg (\neg p \wedge \neg q)), \langle p, i \rangle, \langle  \neg (\neg p \wedge \neg q), j \rangle \}$, for some two different $i, j \in \N$.

\medskip

\n First, we will decompose the expressions contained in $\phi_{1.1}$. Applying the rule $(\text{\R}_{\neg \neg})$ to $\neg \neg (\neg p \wedge \neg q)$ in $\phi_{1.1}$,  and next $(\text{\R}_{\wedge})$ to $\neg p \wedge \neg q$ in $\phi_{1.1.1}$ we obtain the two subsequent branches:

\medskip
\n $\phi_{1.1.1}$: $\phi_{1.1}$, $\{p, \sim ( p  \rightarrow \neg (\neg p \wedge \neg q)), \neg ( p  \rightarrow \neg (\neg p \wedge \neg q)), p, \neg \neg (\neg p \wedge \neg q), \neg p \wedge \neg q \}$

\medskip
\n $\phi_{1.1.1.1}$: $\phi_{1.1.1}$,  $\{p, \sim ( p  \rightarrow \neg (\neg p \wedge \neg q)), \neg (p  \rightarrow \neg (\neg p \wedge \neg q)), p, \neg \neg (\neg p \wedge \neg q), \neg p \wedge \neg q, \neg p, \neg q\}$ --- this is an \textit{almost}  closed branch, since it contains $p$ and $\neg p$, which by the rule $(\text{\R}_{\neg})$ results in the closed branch, because its  last element contains $p$ and $\sim p$ (by definition \ref{Branch: closed/open,complete}):

\medskip
\n $\phi_{1.1.1.1.1}$: $\phi_{1.1.1.1}$, $\{p, \sim ( p  \rightarrow \neg (\neg p \wedge \neg q)), \neg ( p  \rightarrow \neg (\neg p \wedge \neg q)), p, \neg \neg (\neg p \wedge \neg q), \neg p \wedge \neg q, \neg p, \neg q, \sim p\}$.
\medskip

\n Note again that no tableau rule can be applied to the last element of closed branch,  by the definition of \ref{Rule}.

\n Independently, we can extend $\phi_{1.2}$, using the rule $(\text{\R}_{\neg \rightarrow(2)})$ to $\langle p, i \rangle, \langle  \neg (\neg p \wedge \neg q), j \rangle$ (where once $k = i$ and once $k = j$) and thus obtaining:

\medskip

\n $\phi_{1.2.1}$: $\phi_{1.2}$, $\{p, \sim ( p  \rightarrow \neg (\neg p \wedge \neg q)), \neg ( p  \rightarrow \neg (\neg p \wedge \neg q)), \langle p, i \rangle, \langle  \neg (\neg p \wedge \neg q), j \rangle, \langle \sim p, j \rangle \}$ 

\medskip
\n and:

\medskip
\n $\phi_{1.2.2}$: $\phi_{1.2}$, $\{p, \sim ( p  \rightarrow \neg (\neg p \wedge \neg q)), \neg ( p  \rightarrow \neg (\neg p \wedge \neg q)), \langle p, i \rangle, \langle  \neg (\neg p \wedge \neg q), j \rangle,  \langle  \sim  \neg (\neg p \wedge \neg q), i \rangle \}$. 

\medskip

\n The branch $\phi_{1.2}$ can be also extended by use of $(\text{\R}_{i})$  to $\langle  \neg (\neg p \wedge \neg q), j \rangle$: 

\medskip
\n $\phi_{1.2.3}$: $\phi_{1.2}$, $\{p, \sim ( p  \rightarrow \neg (\neg p \wedge \neg q)), \neg ( p  \rightarrow \neg (\neg p \wedge \neg q)), \langle p, i \rangle, \langle  \neg (\neg p \wedge \neg q), j \rangle, \langle p, j \rangle,  \langle q, j \rangle \}$. 

\medskip

First, we will decompose $\phi_{1.2.1}$. The branch $\phi_{1.2.1}$ can be extended by the use of $(\text{\R}_{i})$ to the expression $\langle  \neg (\neg p \wedge \neg q), j \rangle$:  
\medskip

\n $\phi_{1.2.1.1}$: $\phi_{1.2.1}$, $\{p, \sim ( p  \rightarrow \neg (\neg p \wedge \neg q)), \neg ( p  \rightarrow \neg (\neg p \wedge \neg q)), \langle p, i \rangle, \langle  \neg (\neg p \wedge \neg q), j \rangle, \langle \sim p, j \rangle, \langle p, j \rangle, \langle q, j \rangle \}$ --- it is a closed branch

\medskip 

\n The branch $\phi_{1.2.1}$ can also be  extended by the use of $(\text{\R}_{\neg \rightarrow(2)})$ to the expressions $\langle p, i \rangle$, $\langle  \neg (\neg p \wedge \neg q), j \rangle$:  
\medskip

\n $\phi_{1.2.1.2}$: $\phi_{1.2.1}$, $\{p, \sim ( p  \rightarrow \neg (\neg p \wedge \neg q)), \neg ( p  \rightarrow \neg (\neg p \wedge \neg q)), \langle p, i \rangle, \langle  \neg (\neg p \wedge \neg q), j \rangle, \langle \sim p, j \rangle, \langle \sim  \neg (\neg p \wedge \neg q), i \rangle \}$ 

\medskip
\n And now we can extend it by applying $(\text{\R}_{i})$ to $\langle  \neg (\neg p \wedge \neg q), j \rangle$: 

\medskip
\n $\phi_{1.2.1.2.1}$: $\phi_{1.2.1.2}$, $\{p, \sim ( p  \rightarrow \neg (\neg p \wedge \neg q)), \neg ( p  \rightarrow \neg (\neg p \wedge \neg q)), \langle p, i \rangle, \langle  \neg (\neg p \wedge \neg q), j \rangle, \langle \sim p, j \rangle, \langle \sim  \neg (\neg p \wedge \neg q), i \rangle,  \langle p, j \rangle, \langle  q, j \rangle\}$ --- it is a closed branch.

\medskip

\n But also the branch $\phi_{1.2.1.2}$ can be extended by the rule $(\text{\R}_{\sim i})$ applied to $\langle \sim  \neg (\neg p \wedge \neg q), i \rangle$:

\medskip

\n $\phi_{1.2.1.2.2}$: $\phi_{1.2.1.2}$, $\{p, \sim ( p  \rightarrow \neg (\neg p \wedge \neg q)), \neg ( p  \rightarrow \neg (\neg p \wedge \neg q)), \langle p, i \rangle, \langle  \neg (\neg p \wedge \neg q), j \rangle, \langle \sim p, j \rangle, \langle \sim  \neg (\neg p \wedge \neg q), i \rangle, \langle \sim p , i \rangle \}$ 

\medskip

\n $\phi_{1.2.1.2.3}$: $\phi_{1.2.1.2}$, $\{p, \sim ( p  \rightarrow \neg (\neg p \wedge \neg q)), \neg ( p  \rightarrow \neg (\neg p \wedge \neg q)), \langle p, i \rangle, \langle  \neg (\neg p \wedge \neg q), j \rangle, \langle \sim p, j \rangle, \langle \sim  \neg (\neg p \wedge \neg q), i \rangle, \langle \sim q , i \rangle \}$ 

\medskip

Now, to the expression  $\langle  \neg (\neg p \wedge \neg q), j \rangle$, present in both branches we apply $(\text{\R}_{i})$, and get two closed branches:

\medskip

\n $\phi_{1.2.1.2.2.1}$: $\phi_{1.2.1.2}$, $\{p, \sim ( p  \rightarrow \neg (\neg p \wedge \neg q)), \neg ( p  \rightarrow \neg (\neg p \wedge \neg q)), \langle p, i, \rangle, \langle  \neg (\neg p \wedge \neg q), j \rangle, \langle \sim p, j \rangle, \langle \sim  \neg (\neg p \wedge \neg q), i \rangle, \langle \sim p , i \rangle,  \langle p, j \rangle, \langle q, j \rangle \}$ 

\medskip

\n $\phi_{1.2.1.2.3.1}$: $\phi_{1.2.1.2}$, $\{p, \sim ( p  \rightarrow \neg (\neg p \wedge \neg q)), \neg ( p  \rightarrow \neg (\neg p \wedge \neg q)), \langle p, i, \rangle, \langle  \neg (\neg p \wedge \neg q), j \rangle, \langle \sim p, j \rangle, \langle \sim  \neg (\neg p \wedge \neg q), i \rangle, \langle \sim p , i \rangle,  \langle p, j \rangle, \langle q, j \rangle \}$ 

\medskip

Now, we list the possible extensions of the branch $\phi_{1.2.2}$. When we use the rule $(\text{\R}_{\neg \rightarrow(2)})$ to $\langle p, i \rangle, \langle  \neg (\neg p \wedge \neg q), j \rangle$ (where $k = j$ ), we obtain the branch:

\medskip 
\n $\phi_{1.2.2.1}$: $\phi_{1.2.2}$, $\{p, \sim ( p  \rightarrow \neg (\neg p \wedge \neg q)), \neg ( p  \rightarrow \neg (\neg p \wedge \neg q)), \langle p, i \rangle, \langle  \neg (\neg p \wedge \neg q), j \rangle,  \langle  \sim  \neg (\neg p \wedge \neg q), i \rangle, \langle \sim p, j\rangle \}$. 
\medskip 

\n The $\phi_{1.2.2.1}$ we can extend by applying $(\text{\R}_{i})$ to $\langle  \neg (\neg p \wedge \neg q), j \rangle$ Then, we have the closed branch:

\medskip

\n $\phi_{1.2.2.1.1}$: $\phi_{1.2.2.1}$, $\{p, \sim ( p  \rightarrow \neg (\neg p \wedge \neg q)), \neg ( p  \rightarrow \neg (\neg p \wedge \neg q)), \langle p, i \rangle, \langle  \neg (\neg p \wedge \neg q), j \rangle,  \langle  \sim  \neg (\neg p \wedge \neg q), i \rangle, \langle \sim p, j\rangle, \langle  p, j\rangle, \langle q, j \rangle \}$. 
\medskip

\n Alternatively, we can extend $\phi_{1.2.2.1}$ by applying $(\text{\R}_{\sim i})$ to  $\langle  \sim  \neg (\neg p \wedge \neg q), i \rangle$. Then we obtain:
\medskip

\n $\phi_{1.2.2.1.2}$: $\phi_{1.2.2.1}$, $\{p, \sim ( p  \rightarrow \neg (\neg p \wedge \neg q)), \neg ( p  \rightarrow \neg (\neg p \wedge \neg q)), \langle p, i \rangle, \langle  \neg (\neg p \wedge \neg q), j \rangle,  \langle  \sim  \neg (\neg p \wedge \neg q), i \rangle, \langle \sim p, i \rangle \}$. 
\medskip

\n and

\medskip
\n $\phi_{1.2.2.1.3}$: $\phi_{1.2.2.1}$, $\{p, \sim ( p  \rightarrow \neg (\neg p \wedge \neg q)), \neg ( p  \rightarrow \neg (\neg p \wedge \neg q)), \langle p, i \rangle, \langle  \neg (\neg p \wedge \neg q), j \rangle,  \langle  \sim  \neg (\neg p \wedge \neg q), i \rangle, , \langle \sim q, i \rangle \}$. 

\medskip
\n The $\phi_{1.2.2.1.2}$ is obviously closed, while $\phi_{1.2.2.1.3}$ can be extended by $(\text{\R}_{i})$ used to 
$\langle  \neg (\neg p \wedge \neg q), j \rangle$, which also results in a closed branch:

\medskip 
\n $\phi_{1.2.2.1.3.1}$: $\phi_{1.2.2.1.3}$, $\{p, \sim ( p  \rightarrow \neg (\neg p \wedge \neg q)), \neg ( p  \rightarrow \neg (\neg p \wedge \neg q)), \langle p, i \rangle, \langle  \neg (\neg p \wedge \neg q), j \rangle,  \langle  \sim  \neg (\neg p \wedge \neg q), i \rangle, , \langle \sim q, i \rangle, \langle p, j \rangle, \langle q, j \rangle \}$. 

\bigskip

The branch $\phi_{1.2.2}$ can be extended by applying $(\text{\R}_{ i})$ to $ \langle  \neg (\neg p \wedge \neg q), j \rangle$. Then we obtain  the branch:

\medskip

\n $\phi_{1.2.2.2}$: $\phi_{1.2.2}$, $\{p, \sim ( p  \rightarrow \neg (\neg p \wedge \neg q)), \neg ( p  \rightarrow \neg (\neg p \wedge \neg q)), \langle p, i \rangle, \langle  \neg (\neg p \wedge \neg q), j \rangle,  \langle  \sim  \neg (\neg p \wedge \neg q), i \rangle, \langle p, j \rangle, \langle q, j \rangle \}$.

\medskip

The branch $\phi_{1.2.2.2}$ can be extended by applying $(\text{\R}_{\sim i})$ to  $\langle  \sim  \neg (\neg p \wedge \neg q), i \rangle$. Then we obtain two branches:

\medskip

\n $\phi_{1.2.2.2.1}$: $\phi_{1.2.2.2}$, $\{p, \sim ( p  \rightarrow \neg (\neg p \wedge \neg q)), \neg ( p  \rightarrow \neg (\neg p \wedge \neg q)), \langle p, i \rangle, \langle  \neg (\neg p \wedge \neg q), j \rangle,  \langle  \sim  \neg (\neg p \wedge \neg q), i \rangle, \langle p, j \rangle, \langle q, j \rangle, \langle \sim p, i \rangle \}$. 

\medskip

\n and

\medskip

\n $\phi_{1.2.2.2.2}$: $\phi_{1.2.2.2}$, $\{p, \sim ( p  \rightarrow \neg (\neg p \wedge \neg q)), \neg ( p  \rightarrow \neg (\neg p \wedge \neg q)), \langle p, i \rangle, \langle  \neg (\neg p \wedge \neg q), j \rangle,  \langle  \sim  \neg (\neg p \wedge \neg q), i \rangle, \langle p, j \rangle, \langle q, j \rangle, \langle \sim q, i \rangle  \}$. 
\medskip

$\phi_{1.2.2.2.1}$ is closed, while $\phi_{1.2.2.2.2}$  can be still extended by the rule $(\text{\R}_{\neg \rightarrow(2)})$ applied  to $\langle p, i \rangle, \langle  \neg (\neg p \wedge \neg q), j \rangle$ (where $k = j$ ), resulting in the closed branch:

\medskip

\n $\phi_{1.2.2.2.2.1}$: $\phi_{1.2.2.2.2}$, $\{p, \sim ( p  \rightarrow \neg (\neg p \wedge \neg q)), \neg ( p  \rightarrow \neg (\neg p \wedge \neg q)), \langle p, i \rangle, \langle  \neg (\neg p \wedge \neg q), j \rangle,  \langle  \sim  \neg (\neg p \wedge \neg q), i \rangle, \langle p, j \rangle, \langle q, j \rangle, \langle \sim q, i \rangle, \langle \sim p, j\rangle  \}$. 

\medskip

The branch $\phi_{1.2.2.2}$ can be one more extended by the rule $(\text{\R}_{\neg \rightarrow(2)})$ applied  to $\langle p, i \rangle, \langle  \neg (\neg p \wedge \neg q), j \rangle$ (where $k = j$ ), which results in the closed branch:

\medskip

\n $\phi_{1.2.2.2.3}$: $\phi_{1.2.2.2}$, $\{p, \sim ( p  \rightarrow \neg (\neg p \wedge \neg q)), \neg ( p  \rightarrow \neg (\neg p \wedge \neg q)), \langle p, i \rangle, \langle  \neg (\neg p \wedge \neg q), j \rangle,  \langle  \sim  \neg (\neg p \wedge \neg q), i \rangle, \langle p, j \rangle, \langle q, j \rangle, \langle \sim p, j \rangle \}$.

Finally, the branch $\phi_{1.2.3}$ can be decomposed by application of the rule $(\text{\R}_{\neg \rightarrow(2)})$ applied  to $\langle p, i \rangle, \langle  \neg (\neg p \wedge \neg q), j \rangle$, either under the assumption that $k = i$, or under the assumption that $k = j$. In the former case we obtain:

\medskip

\n $\phi_{1.2.3.1}$: $\phi_{1.2.3}$, $\{p, \sim ( p  \rightarrow \neg (\neg p \wedge \neg q)), \neg ( p  \rightarrow \neg (\neg p \wedge \neg q)), \langle p, i \rangle, \langle  \neg (\neg p \wedge \neg q), j \rangle, \langle p, j \rangle,  \langle q, j \rangle, \langle \sim \neg (\neg p \wedge \neg q), i \rangle \}$. 

\medskip

In the latter case:

\medskip

\n $\phi_{1.2.3.2}$: $\phi_{1.2.3}$, $\{p, \sim ( p  \rightarrow \neg (\neg p \wedge \neg q)), \neg ( p  \rightarrow \neg (\neg p \wedge \neg q)), \langle p, i \rangle, \langle  \neg (\neg p \wedge \neg q), j \rangle, \langle p, j \rangle,  \langle q, j \rangle, \langle \sim p, j \rangle \}$. 

\medskip

\n The branch $\phi_{1.2.3.2}$ is closed, while $\phi_{1.2.3.1}$ is extendable in two ways. First, we can apply $(\text{\R}_{\sim i})$ to $\langle \sim \neg (\neg p \wedge \neg q), i \rangle$. Then we get two branches:

\medskip
\n $\phi_{1.2.3.1.1}$: $\phi_{1.2.3.1}$, $\{p, \sim ( p  \rightarrow \neg (\neg p \wedge \neg q)), \neg ( p  \rightarrow \neg (\neg p \wedge \neg q)), \langle p, i \rangle, \langle  \neg (\neg p \wedge \neg q), j \rangle, \langle p, j \rangle,  \langle q, j \rangle, \langle \sim \neg (\neg p \wedge \neg q), i \rangle, \langle \sim p, i \rangle \}$

\medskip

\n $\phi_{1.2.3.1.2}$: $\phi_{1.2.3.1}$, $\{p, \sim ( p  \rightarrow \neg (\neg p \wedge \neg q)), \neg ( p  \rightarrow \neg (\neg p \wedge \neg q)), \langle p, i \rangle, \langle  \neg (\neg p \wedge \neg q), j \rangle, \langle p, j \rangle,  \langle q, j \rangle, \langle \sim \neg (\neg p \wedge \neg q), i \rangle, \langle \sim q, i \rangle  \}$. 

\medskip

\n The branch $\phi_{1.2.3.1.1}$ is closed, but $\phi_{1.2.3.1.2}$ can be extended by $(\text{\R}_{\neg \rightarrow(2)})$ applied  to $\langle p, i \rangle, \langle  \neg (\neg p \wedge \neg q), j \rangle$ with $k = j$. After that we have the closed branch:

\medskip

\n $\phi_{1.2.3.1.2.1}$: $\phi_{1.2.3.1.2}$, $\{p, \sim ( p  \rightarrow \neg (\neg p \wedge \neg q)), \neg ( p  \rightarrow \neg (\neg p \wedge \neg q)), \langle p, i \rangle, \langle  \neg (\neg p \wedge \neg q), j \rangle, \langle p, j \rangle,  \langle q, j \rangle, \langle \sim \neg (\neg p \wedge \neg q), i \rangle, \langle \sim q, i \rangle,  \langle \sim p, j \rangle  \}$. 

\medskip

The second way of extending $\phi_{1.2.3.1}$  is to apply $(\text{\R}_{\neg \rightarrow(2)})$ applied  to $\langle p, i \rangle, \langle  \neg (\neg p \wedge \neg q), j \rangle$ with $k = j$, again. The we have the closed branch:

\medskip

\n $\phi_{1.2.3.1.3}$: $\phi_{1.2.3.1}$, $\{p, \sim ( p  \rightarrow \neg (\neg p \wedge \neg q)), \neg ( p  \rightarrow \neg (\neg p \wedge \neg q)), \langle p, i \rangle, \langle  \neg (\neg p \wedge \neg q), j \rangle, \langle p, j \rangle,  \langle q, j \rangle, \langle \sim \neg (\neg p \wedge \neg q), i \rangle, \langle \sim p, j \rangle \}$. 

\medskip

It is a good time to sum up.  Among 33 kind of branches the complete branches (according to the  definition of a complete branch \ref{Branch: closed/open,complete}) are: $\phi_{1.1.1.1.1}$, $\phi_{1.2.1.1}$, $\phi_{1.2.1.2.1}$, $\phi_{1.2.1.2.2.1}$, $\phi_{1.2.1.2.3.1}$, $\phi_{1.2.2.1.1}$, $\phi_{1.2.2.1.2}$, $\phi_{1.2.2.1.3.1}$, $\phi_{1.2.2.2.1}$, $\phi_{1.2.2.2.3}$ ,$\phi_{1.2.2.2.2.1}$,  $\phi_{1.2.3.1.1}$, $\phi_{1.2.3.1.2.1}$,  $\phi_{1.2.3.1.3}$, $\phi_{1.2.3.2}$.  Because all of those complete branches are closed by the definition of a closed branch \ref{Branch: closed/open,complete}, $p  \rightarrow \neg (\neg p \wedge \neg q)$ is a branch consequence of set $\{p \}$ with respect to the given set of tableau rules (in short: $\{ p \} \vartriangleright p  \rightarrow \neg (\neg p \wedge \neg q)$, where $\vartriangleright$ is determined by the given tableau rules), by the definition of the branch consequence relation (\ref{Branch consequence relation}), since $\{p, \sim ( p  \rightarrow \neg (\neg p \wedge \neg q)) \}$ is a finite set.


Which sets of branches are tableaux? According to definition of tableaux \ref{Tableau} all singletons $\{ \phi_{j} \}$, where $0 \leq j \leq 33 $ are tableaux. However, not all two-elementary sets are tableaux. For example $\{ \phi_{0}, \phi_{1} \}$ is not a tableau since $\phi_{0}$ is not $\{ \phi_{0}, \phi_{1} \}$-maximal as $\phi_{0}$ is a subbranch of $\phi_{1}$. But, for example, $\{ \phi_{1.1}, \phi_{1.2} \}$ is a tableau. Branches $\phi_{1.1}$, $\phi_{1.2}$ differ at the third position, but there exists tableau rule $(\text{\R}_{\neg \rightarrow(1)})$  that enables this, as the last condition of definition of tableau \ref{Tableau} states. We have the following complete tableaux (according to the definition of a complete tableau \ref{Complete incomplete tableau}): 

\smallskip

\n $\{\phi_{1.1.1.1.1}, \phi_{1.2.1.1} \}$, $\{\phi_{1.1.1.1.1}, \phi_{1.2.1.2.1}\}$, \newline $\{\phi_{1.1.1.1.1}, \phi_{1.2.1.2.2.1}, \phi_{1.2.1.2.3.1}\}$, $\{\phi_{1.1.1.1.1}, \phi_{1.2.2.1.1}\}$, \newline $\{\phi_{1.1.1.1.1}, \phi_{1.2.2.1.2}, \phi_{1.2.2.1.3.1} \}$, $\{\phi_{1.1.1.1.1}, \phi_{1.2.2.2.1}, \phi_{1.2.2.2.2.1 } \}$, \newline $\{\phi_{1.1.1.1.1}, \phi_{1.2.2.2.3}\}$,
$\{\phi_{1.1.1.1.1}, \phi_{1.2.3.1.1}, \phi_{1.2.3.1.2.1}\}$, 
$\{\phi_{1.1.1.1.1}, \phi_{1.2.3.2}\}$, \newline $\{\phi_{1.1.1.1.1}, \phi_{1.2.3.1.3}\}$,

\smallskip

\n No branch can be added to any of the ten  sets provided that they remain tableaux. Finally, all of them are closed tableaux by definition of a closed tableau \ref{Closed open tableau}.

We would like to identify general relationships between $\vartriangleright$, the existence of a closed tableau and the semantic consequence relation. To ascertain this we will establish tableau theorem \ref{Tableau metatheorem}. We see that $\{p\} \not \models_{\S} p  \rightarrow \neg (\neg p \wedge \neg q)$ (see: \ref{Przyklad podlogiki S}), however $\{ p \} \vartriangleright p  \rightarrow \neg (\neg p \wedge \neg q)$. So, the general conditions that connect $\models_{\S}$ to $\vartriangleright$ must be not fulfilled because the tableau system is not sound in respect to semantics for $\S$.\end{myfont}
\qed{}
\end{ex}

\section{Tableau metatheory facts}\label{sec:5}

We would like to determine the general relationship between $\models_{\M}$, $\vartriangleright_{\TR}$, and the~existence of a closed tableau. But first, we will establish  a series of definitions and facts. Then, we will use them in subsection \ref{subsec:tableau methatheorem} in a proof of tableau metatheorem for propositional logics \ref{Tableau metatheorem}.

\subsection{From relation  $\models_{\M}$  to branch consequence relation $\vartriangleright_{\TR}$}

\begin{defi}[Closure under rules]\label{Closure under rules} Let $X, Y \subseteq \Ex$. A set $Y$ is \emph{a closure of} $X$   \emph{under tableau rules} iff $X \subseteq Y$ and there exists   a complete branch $\phi \colon  K \longrightarrow \Pow(\textsf{Ex})$   such that $Y = \bigcup \{ \phi(i) \mid i \in K \}$.
\end{defi}

For any set of expressions we have at least one closure under rules, but sometimes there can be more closures. This follows from the definitions of: a rule \ref{Rule}, a tableau rule \ref{Tableau rule} and a complete branch \ref{Branch: closed/open,complete}. It is worth noting that any closure under rules (or in short: closure) is an one-element-branch (by \ref{Coro Set of expressions is branch}).

\begin{lem}[On the existence of open and complete branches]\label{On existence of open and complete branch}

Let $X \cup \{ A \}\subseteq \textsf{For}$. If for all finite $Y \subseteq X$ there exists a complete and open branch starting with $Y^{\t} \cup \{ \langle \sim A, \t \rangle \}$, then there exists a closure of $X^{\t}  \cup \{ \langle \sim A, \t  \rangle\}$ under tableau rules that is an open and complete branch.

\end{lem}

\begin{proof}

Assume $X \cup \{ A \}\subseteq \textsf{For}$. $(\ddag)$ Let for all finite $Y \subseteq X$ there exist a complete and open branch starting with $Y^{\t} \cup \{ \langle \sim A, \t \rangle \}$. We have two possibilities: either $X \subseteq \textsf{For}$ is finite or infinite.

If $X \subseteq \textsf{For}$ is finite, then by $(\ddag)$ there exists a complete and open branch $\phi$ starting with $X^{\t} \cup \{ \langle \sim A, \t \rangle \}$. As a consequence of definition \ref{Closure under rules}, $\bigcup \phi$ is a closure under tableau rules. By corollary \ref{Coro Set of expressions is branch}, $\bigcup \phi$  is an open and complete branch.

Now assume $X \subseteq \textsf{For}$ is infinite, so countable. We enumerate all formulae in $X \subseteq \textsf{For}$:

\[ B_{1}, B_{2}, B_{3}, \dots \]

\n Afterwards, we define some subsets of $X^{\t} \cup \{ \langle \sim A, \t \rangle \}$:

\[ \{ \langle B_{1}, \t \rangle, \allowbreak \langle B_{2}, \t \rangle, \dots, \langle B_{n}, \t \rangle, \langle \sim A, \t \rangle \},\]

\n where $n \geq 1$. By assumption by $(\ddag)$, for all $\{ \langle B_{1}, \t \rangle, \allowbreak \langle B_{2}, \t \rangle, \dots, \allowbreak \langle B_{n}, \t \rangle, \allowbreak \langle \sim A, \t \rangle \}$ there exists at least one complete and open branch starting with \allowbreak $\{ \langle B_{1}, \t \rangle, \allowbreak \langle B_{2}, \t \rangle, \dots, \langle B_{n}, \t \rangle, \allowbreak \langle \sim A, \t \rangle \}$. So, now we define a list of sets of such branches:

\[\Phi_{1}, \Phi_{2}, \Phi_{3}, \dots, \]

\n where for any $1 \leq i$, $\Phi_{i}$ is the set of all complete and open branches starting with $\{ \langle B_{1}, \t \rangle, \allowbreak \langle B_{2}, \t \rangle, \dots, \langle B_{i}, \t \rangle, \allowbreak \langle \sim A, \t \rangle \}$.

Thereafter, for any $\Phi_{n}$, where $n \in  \mathbb{N}$,  we define $\overline{\Phi}_{n}$, the set of all closures made of branches that belong to $\Phi_{n}$:

\[ \psi \in \overline{\Phi}_{n} \text{ iff } \text{ for some } \phi \in \Phi_{n}: \psi = \bigcup \phi. \]

\n Any $\overline{\Phi}_{n}$ contains closures that are open, since  in $\Phi_{n}$ there are only complete and open branches.

We define the following sets: for any $n \in \mathbb{N}$, by $\overline{\overline{\Phi}}_{n}$ we define such a subset of $\overline{\Phi}_{n}$ that for all $\phi \in \overline{\overline{\Phi}}_{n}$ there exists such closure $\psi \in \overline{\Phi}_{n+1}$ that $\phi \subseteq \psi$. We claim that $\overline{\overline{\Phi}}_{n}$ is not empty, for any $n \in \mathbb{N}$.

Let us assume   set  $\overline{\Phi}_{n}$. It contains all open closures of set $\{ \langle B_{1}, \t \rangle, \dots, \allowbreak \langle B_{n}, \t \rangle,  \langle \sim A, \t \rangle \}$. Let  $\overline{\Phi_{n}\cup \{\langle B_{n+1}, \t \rangle\}}$ be the set of all closures under rules of sets $\phi \cup \{\langle B_{n+1}, \t \rangle\}$, where $\phi \in \overline{\Phi}_{n}$. Set $\overline{\Phi_{n}\cup \{\langle B_{n+1}, \t \rangle\}}$ is non-empty which follows from the definitions of: rule \ref{Rule}, tableau rule \ref{Tableau rule} and complete branch \ref{Branch: closed/open,complete}. $\overline{\Phi_{n}\cup \{\langle B_{n+1}, \t \rangle\}}$ contains only complete one-element-branches. However, if all branches in $\overline{\Phi_{n}\cup \{\langle B_{n+1}, \t \rangle\}}$ are closed, then by (CE), (CC) and (CS) in definition of tableau rule \ref{Tableau rule}  also all closures in $\overline{\Phi}_{n+1}$ are closed, which contradicts  assumption $(\ddag)$. Therefore, $\overline{\Phi_{n}\cup \{\langle B_{n+1}, \t \rangle\}} \cap \overline{\Phi}_{n+1}$ is non-empty. As a consequence,  $\overline{\overline{\Phi}}_{n}$ is not empty, for any $n \in \mathbb{N}$, as for some $\phi \in \overline{\overline{\Phi}}_{n}$ there is an complete closure $\psi \in \overline{\Phi}_{n+1}$ such that $\phi \subseteq \psi$.

Consequently, we have the following sequence of non-empty sets:

\[ \overline{\overline{\Phi}}_{1}, \overline{\overline{\Phi}}_{2}, \overline{\overline{\Phi}}_{3}, \dots \]

\n that for any $\overline{\overline{\Phi}}_{n}$ and any $\phi \in \overline{\overline{\Phi}}_{n}$ there exists such $\psi \in \overline{\overline{\Phi}}_{n+1}$ that $\phi \subseteq \phi$.

By choosing some $\phi_{1} \in \overline{\overline{\Phi}}_{1}$ we  define a sequence of open and complete closures:

\[(\alpha): \qquad \phi_{1}, \phi_{2}, \phi_{3}, \dots \]

\n For any $n \in \mathbb{N}$ and any  $\phi_{n}$ in $(\alpha)$, $\phi_{n} \in \overline{\overline{\Phi}}_{n}$, $\phi_{n+1} \in \overline{\overline{\Phi}}_{n+1}$ and $\phi_{n} \subseteq \phi_{n+1}$. Finally, we define set $\Sigma$ as the union of the elements of sequence $(\alpha)$: $\bigcup_{i \in \mathbb{N}} \phi_{i}$.

Let us notice that $\Sigma$ is a closure of $X^{\t}  \cup \{ \langle \sim A, \t  \rangle\}$ under tableau rules that is an open and complete branch. As we already know, $\Sigma$ is a branch, an one-element-branch (by \ref{Coro Set of expressions is branch}), and by the construction of sequence $(\alpha)$,  $X^{\t}  \cup \{ \langle \sim A, \t  \rangle\} \subseteq \Sigma$.

If $\Sigma$ is not open, then --- by definition of \ref{Branch inconsistent set of expressions} --- $\Sigma$  includes a finite, b-inconsistent subset. However, it would mean that for some $i \in \mathbb{N}$,  $\phi_{i} \in (\alpha)$ would be a b-inconsistent subset. But it is contradictory, since every $\phi_{i}$ is open.

For a contradiction, let us then assume that $\Sigma$ is not a complete branch. So --- by definition of complete branch \ref{Branch: closed/open,complete} ---  there exists  such a two-element-branch $\Sigma'$ that $\bigcup \Sigma \subset \Sigma'$, where $\bigcup \Sigma = \Sigma$, since $\Sigma$ is an one-element-branch. Because there exists the two-element-branch $\Sigma'$ such, that $\Sigma \subset \Sigma'$, so there must exist a tableau rule $\R \in \textbf{TR}$ and an $n$-tuple $\langle \Sigma, X_{1}, \dots, X_{n} \rangle \in \R$ such that $\bigcup \Sigma' = X_{i}$, where $1 \leq i \leq n$. However, by the condition of closure under cores (CC) (definition of tableau rule \ref{Tableau rule}), there are:  a minimal, finite subset $Y \subseteq \Sigma$,  subsets $Y_{1} \subseteq X_{1}$, \dots, $Y_{n} \subseteq X_{n}$ and $n$-tuple $\langle Y, Y_{1}, \dots, Y_{n} \rangle \in \R$. Since $Y$ is finite, there is such an element of string $(\alpha)$ $\phi_{i}$ that $Y \subseteq \phi_{i}$. Since $\phi_{i}$ is b-consistent --- by the condition of closure under expansion (CE) (definition of tableau rule \ref{Tableau rule}) ---  there is a rule $\langle \phi_{i}, Z_{1}, \dots, Z_{n} \rangle \in \R$.  However it means that $\phi_{i}$ is not a complete branch.

Summing up, $\Sigma$ is a closure of $X^{\t}  \cup \{ \langle \sim A, \t  \rangle\}$ under tableau rules that is an open and complete branch.\end{proof}

It is worth noting that the notions below constitute some generalizations of the \textit{abstract consistency property} proposed by R. Smullyan and applied by M. Fitting to systems of logic defined by possible world's semantics (see: \cite{MF1983}). Here, however, we have generalized possible world's semantics, proposing the concept of general semantic structure (definition \ref{General semantic structure}), which is something more general than the notion of frame in modal logic. Therefore, in the following definitions and facts (the definition of structure generated by branch \ref{Structure generated by branch}, the corollary on open branch generating model \ref{Open branch generates model}, and the definition on the relationship between models and rules \ref{Models sound with respect to rules}) we give a broader look at the problem of abstract consistency property.

Now, we give a general definition of a structure induced by a branch. 

\begin{defi}[Structure generated by branch]\label{Structure generated by branch}   Let $\phi$ be a branch. $\s = \langle \{W_{i}\}_{i\in M}$, $\{ R_{j}\}_{j\in N}, \vartheta \rangle$ is a \textit{structure generated by} $\phi$ iff

\begin{itemize}

\item $M = \{i \mid   \exists \langle A, w_{i}^{k} (x_{1}, \dots x_{k}) \rangle \in \Ex \: (\langle A, w_{i}^{k} (x_{1}, \dots x_{k}) \rangle \in \bigcup \phi) \}$

\item for all $i \in M$, $W_{i}$ is  a maximal subset of:

\smallskip

$\{m \mid   \exists \langle A, w_{i}^{k} (x_{1}, \dots x_{k}) \rangle \in \Ex \: [(\langle A, w_{i}^{k} (x_{1}, \dots x_{k}) \rangle \in \bigcup \phi \text{ or }$

$\langle \sim A, w_{i}^{k} (x_{1}, \dots x_{k}) \rangle \in \bigcup \phi) \text{ and } m = \langle x_{1}, \dots x_{k} \rangle ]\}$

\smallskip

such that: if  $o \equiv l \in \bigcup \phi$, $m' = \langle x_{1}, \dots, o, \dots, x_{k} \rangle$ and $m'' = \langle x_{1}, \dots, l,\dots,  x_{k} \rangle$, then $m' \not \in W_{i}$ or $m'' \not \in W_{i}$.


\item $N = \{j \mid   \exists r_{j}^{k}(x_{i}, \dots, x_{k}) \in \Ex \:  [r_{j}^{k}(x_{i}, \dots, x_{k})\in \bigcup \phi \text{ or }\sim r_{j}^{k}(x_{i}, \dots, x_{k}) \allowbreak \in \bigcup \phi] \}$

\item for all $x_{1}, \dots, x_{k}  \in \mathbb{N}$,  $j \in N$,  $\langle x_{i}, \dots, x_{k} \rangle \in R_{j}$ iff $r_{j}(x_{i}, \dots, x_{k}) \in \bigcup \phi$

\item for all $A \in \Var$, $i \in M$ and  $\langle x_{1}, \dots x_{k} \rangle \in W_{i}$, $\vartheta(A, \langle x_{1}, \dots x_{k} \rangle) = 1$ iff $\langle A, w_{i}^{k} (x_{1}, \dots x_{k}) \rangle \in \bigcup \phi$.


\end{itemize}

\end{defi}

We have an obvious corollary on the relation between structures and open branches.

\begin{coro}[Open branch generates structure]\label{Open branch generates model}  If $\phi$ is an open branch, then there exists a structure $\s$ generated by $\phi$.
\end{coro}

\begin{proof}By definitions of open branch \ref{Branch: closed/open,complete} and structure generated by branch \ref{Structure generated by branch}. \end{proof}

The last definition is on how complete branches  generate a structure that can be extended to model belonging to $\M$ (we assumed in section \ref{Local semantic consequence relation} as $(\maltese\, \maltese)$) and satisfying an initial set of~expressions.

\begin{defi}[Models sound with respect to rules]\label{Models sound with respect to rules} Let $\phi$ be a complete and open branch and $X
  = \{  \langle B, t \rangle \mid  \langle B, t \rangle \in \bigcup \phi \}$.  $\M$ is \textit{sound} with respect to the rules of $\TR$ iff for all structures $\s$ generated by  $\phi$: (a) $\s$ can be extended to a model $\mathfrak{M} \in \textbf{M}$ and (b) if $\langle A, t \rangle \in X$,  then $\mathfrak{M},  x_{1} \models A$, and if $\langle \sim A, t \rangle \in X$,  then $\mathfrak{M}, x_{1} \not \models A$, where $A \in \For$, $t  =  w_{i}^{k} (x_{1})$, and $x_{1} \in W_{k}$.


\end{defi}

\subsection{From a closed tableau to semantic relation $\models_{\M}$}


Now we need to establish a relationship between tableaux and models.

\begin{defi}[Rules sound with respect to models]\label{Rules sound with respect to models} $\TR$ is \textit{sound} with respect to $\M$ iff for all sets $X_{1}, \dots, X_{i} \subseteq \textsf{Ex}$ (where $1 < i$), all models $\m \in \M$, and all rules $\R \in \TR$, if $\langle X_{1}, \dots, X_{i} \rangle \in \R$ and  $\m$ is suitable for $X_{1}$, then $\m$ is suitable for  $X_{j}$, for some $1 < j \leq i$.
\end{defi}

\begin{lem}[On maximal and open branch] \label{Lemma on maximal and open branch}

Let $\TR$ be sound with respect to $\M$. Let  $X \cup\{ A \}  \subseteq \For$ be finite. If $X \not \models A$, then there exists a complete and open branch starting with $\{\langle B, \t \rangle \mid  B \in X \cup \{  \sim A \} \}$.
\end{lem}

\begin{proof}
Let $\TR$ be sound with respect to $\M$. Let  $X \cup\{ A \}  \subseteq \For$ be finite. Assume that  $X \not \models A$. So there exists a  model $\m \in \M$, domain $W_{k}^{\m}$ and point of relativization  $w^{\m} \in W_{k}^{\, \m}$ that $\m, w^{\m} \models X$ and $\m, w^{\m}  \not \models A$, where domain $W_{k}^{\, \m}$ is the k$^{th}$  domain of $\m$ (note $(\maltese\, \maltese)$, section \ref{Local semantic consequence relation}).

 By \ref{Model sutable to a set of expressions} model $\m$ is suitable for $\{\langle B, \t \rangle \mid  B \in X \}$, but not suitable for $\langle A, \t \rangle$, where $f(\circ(\t)) = w^{\m}$.  So, $\m$ is suitable for $\{\langle B, \t \rangle \mid B \in X \cup \{  \sim A \} \}$. For a contradiction let us additionally assume that all complete branches starting with $\{\langle B, \t \rangle \mid  B \in X \cup \{  \sim A \} \}$ are closed.

Obviously, the set of the complete branches starting with $\{\langle B, \t \rangle \mid  B \in X \cup \{  \sim A \} \}$ is not empty which is guaranteed by definition of tableau rule \ref{Tableau rule}, in particular by the condition of being closed under finite sets (CF) and by assumption that  $X \cup\{ A \}$ is finite. By dint of them we can put all the applicable tuples in a sequence that is a branch. Of course, all of the branches must be of a finite length, since they are closed and by \ref{Rule} the rules cannot be applied to b-inconsistent sets.

Moreover, no model can be suitable for a b-inconsistent set of expressions (by \ref{Branch inconsistent set of expressions}). However, since $\TR$ is assumed to be sound with respect to $\M$ and $\m \in \M$, the model $\m$ is also not suitable for the former parts of the branches, particularly to the first element $\{\langle B, \t \rangle \mid  B \in X \cup \{  \sim A \} \}$ of all the branches under consideration. But this is a contradiction.  Therefore at least one complete branch starting with $\{\langle B, \t \rangle \mid  B \in X \cup \{  \sim A \} \}$ must be open. \end{proof}

The last lemma before summarizing all the facts so far is about the relationship between open branches and all possible tableaux. 

\begin{lem}[On open tableaux]\label{Lemma on open tableaux} Let  $X \cup\{A\} \subseteq \For$ be finite. If there is a~complete and open branch starting with $\{\langle B, \t \rangle \mid B \in X \cup \{\sim A \} \}$, then all complete tableaux $\langle X, A, \Psi \rangle$ are open.

\end{lem}

\begin{proof}

Assume some finite $X \cup\{A\} \subseteq \For$.  Our proof is a proof by contraposition. Let at least one complete tableau $\langle X, A, \Psi \rangle$ be closed. By definition of complete tableau \ref{Complete incomplete tableau}, $\Psi$ contains only complete and closed branches that all start with the set $\{\langle B, \t \rangle \mid B \in X \cup \{\sim A \} \}$. Additionally, all of them are of a finite length, since they are closed.

Let us now consider a complete branch $\phi$ starting with $\{\langle B, \t \rangle \mid B \in X \cup \{\sim A \} \}$. Since $\phi$ is a complete branch, no rule can be applied to $\bigcup \phi$ (by definition of complete branch \ref{Complete incomplete tableau}). By induction on the length of the branches contained in $\Psi$ we are aiming to show that $\bigcup \phi$ contains a b-inconsistent subset of expressions, and so $\phi$ is a  closed branch.

First, in the tableau $\langle X, A, \Psi \rangle$ all branches start with the set of expressions $\{\langle B, \t \rangle \mid B \in X \cup \{\sim A \} \}$. Either $\{\langle B, \t \rangle \mid B \in X \cup \{\sim A \} \}$ is b-consistent, or not. Of course, if $\{\langle B, \t \rangle \mid B \in X \cup \{\sim A \} \}$ is b-inconsistent, then $\phi$ is a  closed branch. So let the set of expressions $\{\langle B, \t \rangle \mid B \in X \cup \{\sim A \} \}$ be b-consistent. If so, then some rule was applied to $\{\langle B, \t \rangle \mid B \in X \cup \{\sim A \} \}$ and any branch in $\Psi$ as a second member is $\{\langle B, \t \rangle \mid B \in X \cup \{\sim A \} \} \cup X_{i}$, for some $i \in \mathbb{N}$,  where $X_{i}$ is an output of the rule. Since $\phi$ is a complete branch, $\bigcup \phi$ must contain a similar set to $X_{i}$, because the tableau rules are closed under similarity (CS) (by \ref{Tableau rule}). Of course, either $\{\langle B, \t \rangle \mid B \in X \cup \{\sim A \} \} \cup X_{i}$ is b-consistent, or not. If  $\{\langle B, \t \rangle \mid B \in X \cup \{\sim A \} \} \cup X_{i}$ is b-inconsistent, then $\phi$ is a  closed branch, as $\bigcup \phi$ contains a similar set to $\{\langle B, \t \rangle \mid B \in X \cup \{\sim A \} \} \cup X_{i}$ and  similar sets both are either b-consistent, or both b-inconsistent (by the definition of similar sets of expressions \ref{Similar sets of expressions}). So let the set of expressions $\{\langle B, \t \rangle \mid B \in X \cup \{\sim A \} \} \cup X_{i}$ be b-consistent.

Second, assume that $\bigcup \phi$ contains the set of expressions $Y_{n}$ that is similar to the $n^{th}$--memeber of one of the branches in the tableau $\langle X, A, \Psi \rangle$.  Again, if $Y$ is b-inconsistent, $\bigcup \phi$, as well as $\phi$ is a closed branch. If not, than  $\bigcup \phi$ must contain the set of expressions $Y_{n+1}$ that is similar to the $n^{th}+1$--member of one of the branches in the tableau $\langle X, A, \Psi \rangle$. However, because all branches in the tableau $\langle X, A, \Psi \rangle$ are finite (as closed), so $\bigcup \phi$  either contains (1) a set of expressions that is similar to the  union of members of one of the branches in $\Psi$, or contains (2) the union of  some complete and open, but infinite branch. If (1) holds, then $\bigcup \phi$ is b-consistent, as well as $\phi$ itself. Otherwise, if  (2) holds,  then there exists an infinite, complete and so open branch that starts from the set $\{\langle B, \t \rangle \mid B \in X \cup \{\sim A \} \}$. By the condition of being closed under similarity  (CS) (\ref{Tableau rule}) the copy of the branch should be in the tableau $\langle X, A, \Psi \rangle$, since it is a complete tableau. But $\langle X, A, \Psi \rangle$ is a closed tableau which is a contradiction and cannot contain any open branch.

Therefore any~complete  branch starting with $\{\langle B, \t \rangle \mid B \in X \cup \{\sim A \} \}$ is closed. \end{proof}

\subsection{Tableau Metatheorem on: $\models_{\M}$, $\vartriangleright_{\TR}$, and tableaux}\label{subsec:tableau methatheorem}

\noindent The former definitions and lemmas are sufficient to conclude the tableau metatheorem for any propositional logics defined by way proposed here.

\begin{theo}[Tableau Metatheorem]\label{Tableau metatheorem} Let (1) the set of rules $\TR$ be sound with respect to the class of models   $\M$ and (2) the class of models $\M$ be sound with respect to the  rules of $\TR$.  Then for all $X \subseteq \For$, $A \in \For$ the following statements are equivalent:

\begin{itemize}

\item $X \models_{\M} A$

\item  $X \vartriangleright_{\TR} A$

\item there is a finite $Y \subseteq X$ and a closed tableau $\langle Y, A, \Phi \rangle$.

\end{itemize}

\end{theo}

\begin{proof}

We assume points (1), (2) and take some $X \cup \{A \}  \subseteq \For$.

(a) $X \models_{\M} A \Longrightarrow X \vartriangleright_{\TR} A$.

\n By contraposition, we assume that $X \not \vartriangleright_{\TR} A$. So, by the definition of the  branch consequence relation \ref{Branch consequence relation} for all finite  $Y \subseteq X$ there exists a complete and open branch starting with $Y^{\t} \cup \{ \langle \sim A, \t \rangle \}$. Hence, by  lemma \ref{On existence of open and complete branch} there exists some closure $Z$ of $X^{\t}  \cup \{ \langle \sim A, \t  \rangle\}$ under the rules that is an open and complete branch. By corollary \ref{Open branch generates model} any open branch generates a structure, so also set $Z$ does. Since we assumed that (2) models were sound with respect to the rules, the structure can be extended to model $\mathfrak{M} \in \textbf{M}$ and (b) if $\langle A, \t \rangle \in X$,  then $\mathfrak{M},  x_{1} \models A$, and if $\langle \sim A, \t \rangle \in X$,  then $\mathfrak{M}, x_{1} \not \models A$, where $A \in \For$, $\t  =  w_{i}^{k} (x_{1})$, and $x_{1} \in W_{k}$. 
Since $X^{\t}  \cup \{ \langle \sim A, \t  \rangle\} \subseteq Z$, $X \not \models_{\M} A$.

(b) $X \vartriangleright_{\TR} A \Longrightarrow$ there is a finite $Y \subseteq X$ and a closed tableau $\langle Y, A, \Phi \rangle$.

\n By contraposition we assume that for all finite $Y \subseteq X$ no tableau $\langle Y, A, \Phi \rangle$ is complete and closed. Obviously, for all finite $Y \subseteq X$ there exists a complete tableau $\langle Y, A, \Phi \rangle$, since for any finite set of expressions there exists a complete branch  (by corollary \ref{Coro Set of expressions is branch} and the definition of tableau rules \ref{Tableau rule}) and so a set of branches being a tableau can be constructed (by \ref{Tableau}). Hence one of the branches that belong to $\Phi$  is complete and open.

(c) There is a finite $Y \subseteq X$ and a closed tableau $\langle Y, A, \Phi \rangle$ $ \Longrightarrow X \models_{\M} A$.

\n Again, by contraposition, we assume that $X \not \models_{\M} A$. So, by lemma \ref{Lemma on maximal and open branch}, there exists a complete and open branch starting with $\{\langle B, \t \rangle \mid B \in X \cup \{  \sim A \} \}$. Hence, for any finite $Y \subseteq X$ there exists a complete and open branch starting with $\{\langle B, \t \rangle \mid B \in Y \cup \{  \sim A \}\}$, since the tableau rules are closed under expansion (condition (CE) in \ref{Tableau rule}). So if for no finite $Y \subseteq X$ there  exists a complete and open branch starting with $\{\langle B, \t \rangle \mid B \in Y \cup \{  \sim A \}\}$, then there would not be one for $\{\langle B, \t \rangle \mid B \in X \cup \{  \sim A \}\}$ either.

Finally, by \ref{Lemma on open tableaux}, all complete tableaux $\langle Y, A, \Psi \rangle$ are open, for all finite $Y \subseteq X$.\end{proof}

\section{Applications of the metatheory}\label{sec:6}

The theory presented may – at first glance – seem quite complex, especially when compared to more popular and didactic approaches to tableaux. But you only need to set it up once, and once you've done that, it shows its value in the ease with which we can generate specific tableau systems. Consider, for example, how we might define a tableau system for some semantically formulated propositional logic.

We start with $\models$ logic specified by some set of models. Now the models need to be linked to the tableau language using a fitting function. Then, you must provide a set of tableau rules that meet the conditions in the tableau rule definition \ref{Tableau rule}. Then, according to the tableau metatheorem \ref{Tableau metatheorem}, if it is positively established that the tableau models and rules are sound with respect to each other, an adequate tableau system will thus be established. Thus, the metatheory presented in this paper offers a simplification of the process of defining all concepts and proving specific facts when constructing the $\vartriangleright$ equivalent of a given logic $\models$. In fact, it is enough to define the rules and check their relationship with the models, so in the presented methodology, the construction of an adequate tableau system comes down to three steps. Now, we show these three steps using the example of multimodal, three-valued logic. 

\begin{ex}\begin{myfont} 
Here we will show an application of the metatheory to the construction of adequate tableau system for some bi-modal, three-valued logic. 

Let us assume the set of formulae $\For$ builit with variables, four binary connectives: $\wedge$, $\vee$, $\rightarrow$, $\leftrightarrow$,  three unary connectives: negation $\neg$,  modal, alethic operator: $\Diamond$, and one modal, epistemic operator: $\K$. So, this satisfies the definition of formula \ref{Formula}.

The set of structures for the language, we initially consider,    consists of all of the ordered tuples: $\langle W, R^\K, R^\Diamond, \vartheta\rangle$, where $W\neq \varnothing$,
$R^\K\subseteq W^2$, $ R^\Diamond \subseteq W^2$, $\vartheta \colon W\times \textsf{Var} \longrightarrow \{1, 0, \frac{1}{2}\}$. This set can be identified with some subset of general semantic structures $\langle \{W_{i}\}_{i\in M}$, $\{ R_{j}\}_{j\in N}$ $\vartheta' \rangle$ from  $\GS$ (definition \ref{General semantic structure}); exactly with the structures: $\langle \{W_{1}, W_{2}, W_{3}\}$, $\{ R_{1}, R_{2}, R_{3}\}$, $\vartheta' \rangle$, where $W_{1} = W$, $W_{2} = \{1, 0, \frac{1}{2}\}$, $W_{3} = W_{1} \times W_{2}$, $R_{1} = R^\K$, $R_{2} = R^\Diamond$, $R_{3} = \leq \subseteq \{1, 0, \frac{1}{2} \} \times \{1, 0, \frac{1}{2} \}$ and the function $\vartheta' \colon (W_{1} \cup W_{2} \cup W_{3}) \times \Var \longrightarrow \{ 0, 1 \}$ is defined by the condition: for all $p_{i} \in \Var$, for all $w \in W_{1}$ there exists exactly one such $x \in W_{2}$ that $\vartheta'(\langle w, x \rangle, p_{i}) = 1$, and $\vartheta'(\langle p_{i}, w \rangle ) = 1$ iff $\vartheta'(\langle p_{i}, w, 1 \rangle ) = 1$, while for the rest of elements $\vartheta'$ assigns 0. Of course, the k$^{th}$ domain of structure is $W_{1}$ or $W$ (the definition of  domain structure \ref{kth domain}). Now, we must define a satisfiability relation to transform the structures into models $\M$ (see the definition of models \ref{Models of For}).
We do that in two steps. Having a structure $\langle W, R^\K, R^\Diamond, \vartheta\rangle$, we extend $\vartheta$ to the function $V \colon W \times \For \longrightarrow \{1, 0, \frac{1}{2}\}$ in the following way,\label{From v to V for Diamond K 3} for all $w, u \in W$: 

\begin{align*}
V(w,\neg A) &= 1-V(w,A) \\
V(w, A\wedge B) &= min\{V(w,A), V(w,B)\}\\
V(w, A\lor B) &= max\{V(w,A), V(w,B)\}\\
V(w, A\to B) &= min\{1-V(w,A)+V(w,B),1\}\\
V(w, \K A) &= min\{V(u,A): wR^\K u\}\\
V(w, \Diamond A) &= max\{V(u,A): w R^\Diamond u\}.
\end{align*}

\n Naturally, the functions $min$ and $max$ are defined on the arithmetic order $\leq$, and $min (\emptyset)= 1$, while $max (\emptyset)= 0$.  It is worth noting that the conditions for non-modal connectives are identical to Łukasiewicz's three-valued logic $\bold{\L}3$, however extended to worlds. 

Among the structures, we select those $\s = \langle W, R^\K, R^\Diamond, \vartheta\rangle$ that meet the following condition: (Reflexivity)  $\forall_{w \in W} R^\K (w, w)$\footnote{This condition is usually assumed in the epistemic interpretation of modality, because it is related to the validity of $\K A \rightarrow A$.}.
The set of structures that satisfy these conditions we denote by $\S$. 

For any structure $\s = \langle \{W_{1}, W_{2}, W_{3}\}$, $\{ R_{1}, R_{2}, R_{3}\}$, $\vartheta' \rangle$ (aka $\langle W, R^\K, R^\Diamond, \vartheta\rangle$) the satisfiability relation $\models_{\s} \subseteq (W_{1} \cup W_{2} \cup W_{3}) \times \For$ is now defined as follows:  $\langle w, n  \rangle \models_{\s}  A$ iff $V(w, A) = n$, for $n \in \{1, 0, \frac{1}{2}\}$, and $w \models_{\s}  A$ iff $V(w, A) = 1$. So, the class $\S$ under the union of $\models_{\s}$ transforms into  the class of models $\M$, and thus in a model $\m$ and in a world $w^{\m}$: $\mathfrak{M},w\vDash A$ iff $V(w, A) = 1$. Finally, by the definition of semantic consequence relation  \ref{Semantic consegence relation} and the class of models $\M$ we obtain semantically determined logic $\models_{\K\Diamond 3}$. The question is how to define an adequate tableau system?

The first step of constructing an adequate tableau system to $\models_{\K\Diamond 3}$ is to link models $\M$ to the tableau language using a fitting function. So, let us start describing a tableau language. As labels (see the definition of tableau labels \ref{Tableu labels}) we assume the union of $\N$ and the set of pairs $\N \times \{1, 0, \frac{1}{2}\}$ which we denote by $\TE$. We omit  functional symbols: the elements from $\N$ denote objects from $W_{1}$, the pairs from $\N \times \{1, 0, \frac{1}{2}\}$ denote objects from $W_{3}$ --- so this is the assignment under the fitting function. (Of course, we could, for example,  use a symbol $w_{3}$ to form labels, $w_{3}(i, n)$, or $w_{1}$ for $w_{1}(i)$,  but it seems redundant in this case.) Now, we assume the set of  expressions (see the definition of expression: \ref{Expression}) as $\Ex = \{ \langle A,  t \rangle  | A \in \For, t \in \TE \} \cup \{ \langle \sim A, t \rangle  | A \in \For, t \in \TE \}$. When it is clear, we just write: $A, t $ instead of $\langle A , t\rangle$ and $ \sim A, t $ instead of $\langle \sim A, t \rangle$. We also assume that a label starting a proof $\t$ must have the form $i \in \N$. So, the proof of the formula $A$ on the ground of the set of formulae $Y$ starts from $Y^{\t}\cup\{\langle \sim A, \t \rangle\}$, for some $\t$ (look at the remark on the starting label in  the section \ref{Remarks on starting label}) $Y^{\t}\cup\{\langle \sim A, \t \rangle\}$.\label{Remarks on starting label for KDiamond3}

The second step is to provide a set of tableau rules that meet the conditions in the tableau rule definition \ref{Tableau rule}.
This step is creative and requires a proposal. Our set of tableau rules $\TR_{\K \Diamond 3}$ contains:


\begin{align*}
(\text{\R}_{i1})\quad\begin{array}{c}  A, i\\ \hline\\[-1,08em] A, i, 1  \end{array}\quad\quad
(\text{\R}_{\sim i})\quad\begin{array}{c} \sim A, i\\ \hline\\[-1,08em] \begin{array}{l|c}  A, i, 0  & A, i, \frac{1}{2}\end{array}\end{array}\quad\quad(\text{\R}_{\sim})\quad\begin{array}{c}  A, i, n \\  A, i, m  \\ \hline\\[-1,08em]  \sim A, i, n  \end{array}
\end{align*}

\n When we start the proof with $Y^{\t}\cup\{\langle \sim A, \t \rangle\}$, for some $\t$, the formula $A$, and the set of formulae $Y$,  the rule $(\text{\R}_{i1})$ allows us to state that the formulae in $Y$ have the logical value 1 in $i$, while the rule $(\text{\R}_{\sim} i)$  brings the information that $A$ has got in the world $i$ either $0$ or $\frac{1}{2}$. The rule $(\text{\R}_{\sim})$ is simply to close a branch according to the paradigm when it contains $A, i, n$ and $ A, i, m$, where $n $ and $m$ are different. So, the clause for $(\text{\R}_{\sim})$ is that $n$ and $m$ are different. 

There are two things worth noting about these three tableau rules.

(a) The adoption of  the rules  ($\text{\R}_{i1}$) and ($\text{\R}_{\sim i}$)  is justified by the general approach presented in our paper, which requires that we start the proof with the assumption that the premises are satisfied in the $k_{th}$ element of domain $W_{k}$, while the conclusion is not satisfied in that element (look at the definition of k$^{th}$ domain structure \ref{kth domain}). However, since we are analyzing the case of three-valued logic, we need to allow for more sub-options, which we do by adopting rules ($\text{\R}_{i1}$) and ($\text{\R}_{\sim i}$). Rule ($\text{\R}_{i1}$) says that satisfiability in element $i$ means having the logical value 1, while rule ($\text{\R}_{\sim i}$) says that unsatisfiability means having the value 0 or $\frac{1}{2}$.

However, if we assume that we start the tableau proof with labels: $B, i, 1$ for the premises, and $\sim A, i, 1$ for the conclusion, we would need only a rule that directly tells us to consider two options:

\begin{align*}
\begin{array}{c} \sim A, i, 1\\ \hline\\[-1,08em] \begin{array}{l|c}  A, i, 0  & A, i, \frac{1}{2}\end{array}\end{array}
\end{align*}

(b) We could also ignore the rule $(\text{\R}_{\sim})$ if we assumed that the appearance of $A, i, n$ and $A, i, m$ on a branch gives b-inconsistency if $n$ is different from $m$. However, in our metatheory we generally assumed that a branch is b-inconsistent if it contains $A, t$ and $\sim A, t$, for some label $t$ (look at the comment on tableau meta-negation $\sim$ in section \ref{sec:3}). Therefore the expressions $A, i, n$ and $A, i, m$ reduce to $A, i, n$, $\sim A, i, n$ by the rule  $(\text{\R}_{\sim})$.

As is obvious, our tableau systems metatheory is not ``tailor-made'', but is redundantly tailored. But that is precisely what makes it general. So we can start with more tableau rules, and then ``tailor'' them to a more economical form. Here we will stay close to the initial, general approach.

The subsequent rules (from $(\text{\R}_{\neg 1})$ to $(\text{\R}_{\leftrightarrow \frac{1}{2}})$) correspond to \L ukasiewicz's valuations adopted in the semantics.

\begin{align*}
\quad(\text{\R}_{\neg 1})\quad\begin{array}{c} \neg A, i,  1 \\ \hline\\[-1em] A, i,  0
\end{array}\quad\quad(R_{\neg 0})\quad\begin{array}{c} \neg A, i,  0 \\ \hline\\[-1em] A, i,  1
\end{array}\quad\quad(\text{\R}_{\neg \frac{1}{2}})\quad\begin{array}{c} \neg A, i,  \frac{1}{2} \\ [0.2em]\hline\\[-1em] A, i,  \frac{1}{2}
\end{array}
\end{align*}

\begin{align*}
(\text{\R}_{\wedge 1})\quad\begin{array}{c}  A \wedge B, i,  1 \\ \hline\\[-1em] A, i,  1 \\ B, i, 1
\end{array}\quad(\text{\R}_{\wedge 0})\quad\begin{array}{c}A \wedge B, i,  0\\ \hline\\[-1,08em] \begin{array}{l|c}  A, i, 0  & B, i,  0\end{array}\end{array}
\quad(\text{\R}_{\wedge \frac{1}{2}})\quad\begin{array}{c}A \wedge B, i,  \frac{1}{2}\\ [0.2em]\hline\\[-1,08em] \begin{array}{l|c|r}  A, i, 1 &  A, i,  \frac{1}{2} & A, i,  \frac{1}{2} \\[0,17em]   B, i,  \frac{1}{2} & B, i,  \frac{1}{2} & B, i,  1 \end{array}\end{array}
\end{align*}

\begin{align*}
(\text{\R}_{\vee 1})\quad\begin{array}{c}A \vee B, i,  1\\ \hline\\[-1,08em] \begin{array}{l|c}  A, i, 1  & B, i,  1\end{array}\end{array}\quad\quad(\text{\R}_{\vee 0})\quad\begin{array}{c}  A \vee B, i,  0 \\ \hline\\[-1em] A, i,  0 \\[-0,2em] B, i, 0
\end{array}
\quad(\text{\R}_{\vee \frac{1}{2}})\quad\begin{array}{c}A \vee B, i,  \frac{1}{2}\\ [0.2em]\hline\\[-1,08em] \begin{array}{l|c|r}  A, i, 0 & A, i,  \frac{1}{2} & A , i,  \frac{1}{2}  \\ [0.2em] B, i,  \frac{1}{2} & B , i,  \frac{1}{2} & B , i,  0 \end{array}\end{array}
\end{align*}

\begin{align*}
(\text{\R}_{\rightarrow 1})\quad\begin{array}{c}A \rightarrow B, i,  1\\ \hline\\[-1,08em] \begin{array}{l|c|r}  A, i, 0  & B, i,  1 & A, i, \frac{1}{2}\\[0,2em]  & &   B, i,  \frac{1}{2}\end{array}\end{array}\quad\quad(\text{\R}_{\rightarrow 0})\quad\begin{array}{c}  A \rightarrow B, i,  0 \\ \hline\\[-1em] A, i,  1\\[-0,2em] B, i, 0
\end{array}
\quad(\text{\R}_{\rightarrow \frac{1}{2}})\quad\begin{array}{c}A \rightarrow B, i,  \frac{1}{2}\\ [0.2em]\hline\\[-1,08em] \begin{array}{l|c}  A, i, 1  & A, i,  \frac{1}{2} \\ B, i,  \frac{1}{2} & B, i,  0 \end{array}\end{array}
\end{align*}

\begin{align*}
(\text{\R}_{\leftrightarrow 1})\quad\begin{array}{c}A \leftrightarrow B, i,  1\\ \hline\\[-1,08em] \begin{array}{l|c|r}  A, i, 1 & A, i, 0 & A, i, \frac{1}{2} \\ [0,09em]  B, i, 1 & B, i, 0 & B, i,  \frac{1}{2}\end{array}\end{array}
\quad\quad\quad(\text{\R}_{\leftrightarrow 0})\quad\begin{array}{c}  A \leftrightarrow B, i,  0 \\ \hline\\[-1,08em] \begin{array}{l|c}  A, i, 1& A, i,  0 \\  B, i,  0 & B, i,  1 \end{array}\end{array}
\end{align*}
\begin{align*}
(\text{\R}_{\leftrightarrow \frac{1}{2}})\quad\begin{array}{c}A \leftrightarrow B, i,  \frac{1}{2}\\ [0.2em]\hline\\[-1,08em] \begin{array}{c|c|c|c}  A, i, 1 & A, i,  \frac{1}{2} & A, i, 0 &  A, i, \frac{1}{2}\\ B, i,  \frac{1}{2} & B, i,  1 & B, i,  \frac{1}{2} & B, i, 0\end{array}\end{array}
\end{align*}

In all subsequent rules, if the expression $i r^{\Diamond} j$ (or $i r^{\K} j$) appears in the denominator, then $j$ must be new; if in the numerator, it simply already appeared on the branch.

The rules for $\Diamond$ and $\K$ correspond to the meanings we assigned to them in the semantics for $\models_{\K\Diamond 3}$.

\begin{align*}
(\text{\R}_{\Diamond 1})\quad\begin{array}{c} \Diamond A, i, 1 \\ \hline\\[-1em] i r^{\Diamond} j \\ A, j, 1\\
\end{array}\quad(\text{\R}_{\Diamond 0})\quad\begin{array}{c} \Diamond A, i,  0 \\ i r^{\Diamond} j\\ \hline\\[-1em] A, j,  0 
\end{array}\quad(\text{\R}_{\Diamond\frac{1}{2}})\quad\begin{array}{c} \Diamond A, i, \frac{1}{2}  \\ [0.2em]\hline\\[-1em] i r^{\Diamond} j \\  A, j, \frac{1}{2}
\end{array}
\quad(\text{\R}'_{\Diamond\frac{1}{2}})\quad\begin{array}{c} \Diamond A, i, \frac{1}{2}\\ i r^{\Diamond} j   \\[0.2em]\hline\\[-1em] \begin{array}{l|c} A, j, \frac{1}{2} & A, j, 0\end{array}
\end{array}
\end{align*}

\begin{align*}
(\text{\R}_{\K 1})\quad\begin{array}{c} \K A, i, 1 \\ i r^{\K} j \\ \hline\\[-1em]  A, j, 1\\
\end{array}\quad(\text{\R}_{\K 0})\quad\begin{array}{c} \K A, i,  0 \\  \hline\\[-1em] i r^{\K} j\\ A, j,  0 
\end{array}\quad(\text{\R}_{\K\frac{1}{2}})\quad\begin{array}{c} \K A, i, \frac{1}{2}  \\ [0.2em]\hline\\[-1em] i r^{\K} j \\  A, j, \frac{1}{2}
\end{array}
\quad\quad(\text{\R}'_{\K\frac{1}{2}})\quad\begin{array}{c} \K A, i, \frac{1}{2}\\ i r^{\K} j   \\[0.2em]\hline\\[-1em] \begin{array}{l|r} A, j, \frac{1}{2} & A, j, 1\end{array}
\end{array}
\end{align*}

The remaining rules correspond to the property of (Reflexivity). 
Three variants of the rule for (Reflexivity) have as premises expressions in which an index can appear:

\begin{align*}
\quad(\text{\R}_{\text{ref}})\quad\begin{array}{c} A, i, n \\\hline\\[-1em]  i r^{\K} i  \end{array}\quad\quad\quad(\text{\R}_{\text{ref} \,\K})\quad\begin{array}{c} ir^{\K} j\\\hline\\[-1em]  j r^{\K} j  \end{array}\quad\quad\quad(\text{\R}_{\text{ref} \,\Diamond})\quad\begin{array}{c} ir^{\Diamond} j\\\hline\\[-1em]  j r^{\K} j  \end{array}
\end{align*}



\n The set of tableau rules determines the branch consequence relation $\vartriangleright_{\TR_{\K \Diamond 3}}$ (see: \ref{Branch consequence relation}). 

The last  step of construction  is to  establish that the models $\M$ and the rules $\TR_{\K \Diamond 3}$ are sound with respect to each other to be able to apply the Tableau Metatheorem  \ref{Tableau metatheorem}. First we show that the rules are sound with respect to models. 

\medskip

\n \textbf{Fact 6.2.}\label{DK3 Rules sound to models} Rules from $\TR_{\K \Diamond 3}$ are sound with respect to the class of models $\M$.     

\smallskip
\n \textit{Proof.} To show that the rules are sound with respect to models (according to the definition \ref{Rules sound with respect to models}) let us take some sets $X_{1}, \dots, X_{i} \subseteq \textsf{Ex}$ (where $1 < i$), a model $\m \in \M$, and a rule $\R \in \TR_{\K \Diamond 3}$. Then let us assume that $\langle X_{1}, \dots, X_{i} \rangle \in \R$ and  $\m$ is suitable for $X_{1}$. We must check all cases of rules and show that $\m$ is suitable for  $X_{j}$, for some $1 < j \leq i$. Since the model $\m$ is suitable for $X_{1}$, according the definition of model suitable for a set of expressions \ref{Model sutable to a set of expressions} there is the function  $f': \mathbb{I} \longrightarrow \bigcup_{i \in \{1, 2, 3 \}}  W_{i}$, such that the following conditions  are fulfilled:

\begin{itemize}

\item if $\langle A, i \rangle  \in X_{1}$, then $f'(i) \in W$ and $\m, f'(i) \models A$

\item if $\langle \sim A, i \rangle \in X_{1}$, then $f'(i) \in W$ and $\m, f'(i) \not \models A$

\item if $\langle A, i, n \rangle \in X_{1}$, then $f'(i) \in W$, $n \in \{1, 0, \frac{1}{2}\}$ and $\m, f'(i), n \models A$

\item if $\langle \sim A, i, n \rangle \in X_{1}$, then $f'(i) \in W_{1}$, $n \in \{1, 0, \frac{1}{2}\}$ and $\m, f'(i), n \not \models A$

\item if $i r^{\K} j \in X_{1}$, then $\langle f'(i), f'(j) \rangle  \in R^{\K}$

\item if $i r^{\Diamond} j \in X_{1}$, then $\langle f'(i), f'(j) \rangle  \in R^{\Diamond}$.

\end{itemize}

\n We will check two exemplary cases. 

Let $\R = \text{\R}_{\neg 1}$ and assume that $\langle \neg A, i,  1 \rangle \in X_{1}$. So $\m, f'(i), 1 \models \neg A$ and $V(f'(i),\neg A) = 1$. Therefore $1 = 1- V(f'(i),A)$ and $V(f'(i),A) = 0$. As a consequence $\m, f'(i), 0 \models A$. Finally, the model $\m$ is suitable to $X_{1} \cup \{\langle \neg A, i,  0 \rangle \}$. 

Let $\R = (\text{\R}'_{\K\frac{1}{2}})$ and assume that $\langle \K A, i, \frac{1}{2} \rangle, i r^{\K} j \in X_{1}$. So $\m, f'(i), \frac{1}{2} \models  A$  and $\langle f'(i), f'(j) \rangle \in R^{\K}$. So $V(f'(i), \K A) = \frac{1}{2} = min\{ V(u,A): f'(i)R^\K u\}$. Consequently, $V(f'(j), A) \neq 0$, and  
$V(f'(j), A) = \frac{1}{2}$ or $V(f'(j), A) = \frac{1}{2}$. Therefore, $\m, f'(j), \frac{1}{2} \models  A$ or $\m, f'(j), 1 \models  A$. Finally, the model $\m$ is suitable to $X_{1} \cup \{\langle  A, j,  \frac{1}{2} \rangle \}$  or $X_{1} \cup \{\langle A, j,  1 \rangle \}$.


The remaining cases we check similarly. $\Box$

\smallskip
Now  we prove that the models are sound with respect to rules.

\medskip

\n \textbf{Fact 6.3.}\label{DK3 Models sound to rules} The class of  models $\M$ are sound 
with respect to rules from $\TR_{\K \Diamond 3}$.  

\smallskip

\n \textit{Proof.} Let $\phi$ be a complete and open branch. 
By corollary \ref{Open branch generates model} (Open branch generates structure) there exists a structure  generated by $\phi$. By the definition \ref{Structure generated by branch} (Structure generated by branch) the structure has a form $\langle W, W', R^{1}, R^{2}, \vartheta \rangle$:

\begin{itemize}
    \item $i \in W$ iff $\langle A, i \rangle \in  \bigcup \phi$ or $\langle \sim A, i \rangle \in  \bigcup \phi$, where $i \in \N$,

\item a pair $\langle i, n \rangle \in W'$ iff $\langle A, i, n \rangle \in \bigcup \phi$ or $\langle \sim A, i, n \rangle \in  \bigcup \phi$, where  $i \in \N$ and $n \in \{1, 0, \frac{1}{2}\}$,

\item $R^{1}(i, j)$ iff $i r^\K j \in  \bigcup \phi$,

\item $R^{2}(i, j )$ iff $i r^\Diamond j \in  \bigcup \phi$,

\item $\vartheta$ is such function $\vartheta \colon (W \cup W') \times \Var \longrightarrow \{ 0, 1 \}$ that for all $A \in \Var$:
\begin{itemize}

\item $\vartheta(i, n, A) = 1$ iff $\langle A, i, n \rangle \in  \bigcup \phi$,

\item  $\vartheta(i, A) = 1$ iff $\langle A, i \rangle \in  \bigcup \phi$.

\end{itemize}
\end{itemize}

\n We extend the structure to  structure $\langle \{W_{1}, W_{2}, W_{3}\}$, $\{ R_{1}, R_{2}, R_{3}\}$, $\vartheta' \rangle$, where $W_{1} = W \cup \{i \colon \langle i, n \rangle \in W' \}$, $W_{2} = \{1, 0, \frac{1}{2}\}$, $W_{3} = W_{1} \times W_{2}$, $R_{1} = R^{1}$, $R_{2} = R^{2}$, $R_{3} = \leq \subseteq \{1, 0, \frac{1}{2} \} \times \{1, 0, \frac{1}{2} \}$, and the function $\vartheta' \colon (W_{1} \cup W_{2} \cup W_{3}) \times \Var \longrightarrow \{ 0, 1 \}$ is defined by the condition: $\vartheta'(x) = \vartheta(x)$, for all $x \in (W \cup W')\times \Var$, while for the rest of elements $\vartheta'$ assigns 0. Now, we extend $\vartheta'$ to $V'\colon (W_{1} \cup W_{2} \cup W_{3}) \times \Var \longrightarrow \{ 0, 1 \}$ using the conditions from the begging of this example \ref{From v to V for Diamond K 3}, when we define a model for the logic we consider.  It is easy to notice that the structure defined satisfies condition (Reflexivity)  because it was generated by a complete branch to which the rules  $(\text{\R}_{\text{ref}})$, $(\text{\R}_{\text{ref} \,\K})$,  $(\text{\R}_{\text{ref} \,\Diamond})$ have been applied. So, for any $i \in W_{1}$, $R_{1}(i, i)$. 
Now, we define the relation of satisfiability $\models \subseteq (W_{1} \cup W_{2} \cup W_{3}) \times \For$ as follows:  $\langle w, n  \rangle\models A$ iff $V'(w, A) = n$, for $n \in \{1, 0, \frac{1}{2}\}$, and $w\models A$ iff $V'(w, A) = 1$. 

So far we have shown that the generated structure can be extended to a model $\m \in \M$, which is point (a) of the definition \ref{Models sound with respect to rules} (Models sound with respect to rules). By the same definition we need to prove the point (b), that is, if $\langle A, i \rangle \in X$,  then $\mathfrak{M},  i \models A$, and if $\langle \sim A, i\rangle \in X$,  then $\mathfrak{M}, i \not \models A$, where $A \in \For$, because we assumed at the begining \ref{Remarks on starting label for KDiamond3} that the starting label is $i \in \N$. To do this, however, we must generally show that if $\langle A, t \rangle \in X$,  then $\mathfrak{M},  t \models A$, and if $\langle \sim A, t\rangle \in X$,  then $\mathfrak{M}, t \not \models A$, where $A \in \For$, $t = i$ and $i \in W_{1}$,  or $t = \langle i, n \rangle$, $i \in W_{1}$, and $n \in W_{2} = \{1, 0, \frac{1}{2}\}$.  

The proof is inductive on the complexity of the formula $A$. 

For variables we have: if $\langle A, t \rangle \in  \bigcup \phi$,  then $\mathfrak{M},  t \models A$, and if $\langle \sim A, t\rangle \in X$,  then $\mathfrak{M}, t \not \models A$, where $A \in \Var$, $t = i$ and $i \in W_{1}$,  or $t = \langle i, n \rangle$, $i \in W_{1}$, and $n \in W_{2} = \{1, 0, \frac{1}{2}\}$, by the definition of \ref{Structure generated by branch} (Structure generated by branch). For more complex formulae we must consider the cases of application of tableau rules, assuming inductively that if $\langle A, t \rangle \in \bigcup \phi$,  then $\mathfrak{M},  t \models A$, and if $\langle \sim A, t\rangle \in \bigcup \phi$,  then $\mathfrak{M}, t \not \models A$, where $t = i$ and $i \in W_{1}$,  or $t = \langle i, n \rangle$, $i \in W_{1}$, and $n \in W_{2} = \{1, 0, \frac{1}{2}\}$, $A \in \For$ with complexity $m$. So, we must show that $\langle A, t \rangle \in \bigcup \phi$,  then $\mathfrak{M},  t \models A$, and if $\langle \sim A, t\rangle \in \bigcup \phi$,  then $\mathfrak{M}, t \not \models A$, where $t = i$ and $i \in W_{1}$,  or $t = \langle i, n \rangle$, $i \in W_{1}$, and $n \in W_{2} = \{1, 0, \frac{1}{2}\}$, $A \in \For$ with complexity $m+1$.

Let $\langle A, t \rangle \in \bigcup \phi$.  We consider three exemplary cases. 

Let $t = \langle i, \frac{1}{2} \rangle$ and $A = \Diamond B$. Because the branch $\phi$ is complete,  the rule $(\text{\R}_{\Diamond\frac{1}{2}})$ was applied.  So, $\langle B, i,  \frac{1}{2} \rangle, i r^{\Diamond} j \in \bigcup \phi$. But then, also  $(\text{\R}'_{\Diamond\frac{1}{2}})$ has been applied, resulting in $\langle B, k, \frac{1}{2} \rangle \in \bigcup \phi$ or $\langle B, k, 0 \rangle \in \bigcup \phi$, for any such $k \in \N$ that $i r^{\Diamond} k \in \bigcup \phi$. Whatever has happened, by the inductive hypothesis,  $\mathfrak{M},  \langle k, n \rangle \models B$, for any such $k$ and $n \in \{\frac{1}{2}, 0\}$. But $\mathfrak{M},  \langle k, \frac{1}{2}\rangle \models B$, when $k = j$. So, $V'(w, \Diamond B) = max\{V'(u,B): w R^\Diamond u\} = \frac{1}{2}$, and $\mathfrak{M},  \langle i, \frac{1}{2}\rangle \models A$.

Let $t = i$. Because the branch $\phi$ is complete,  the rule $(\text{\R}_{i1})$ was applied and $\langle A, i, 1  \rangle \in \bigcup \phi$. By the inductive hypothesis,  $\mathfrak{M},  \langle i, 1 \rangle \models A$. So, $V'(i, A) = 1$ and $\mathfrak{M}, i \models  A$. 

Let $t = i$ and $A = \sim B$. Because the branch $\phi$ is complete,  the rule $(\text{\R}_{\sim })$ was applied, then  $\langle A, i, 0  \rangle \in \bigcup \phi$ or $\langle A, i, \frac{1}{2}  \rangle \in \bigcup \phi$. By the inductive hypothesis,  $\mathfrak{M},  \langle i, 0 \rangle \models A$ or $\mathfrak{M},  \langle i, \frac{1}{2}  \rangle \rangle \models A$. So, $V'(i, A) \neq 1$ and $\mathfrak{M}, i \not \models  A$. 

The remaining cases we check similarly. $\Box$

Finally we can apply our Tableau Metatheorem \ref{Tableau metatheorem}, since we have proved the facts on soundness
of models and rules 
 (6.2, 6.3), and automatically get the adequacy of $\models_{\Diamond \K 3}$ and $\TR_{\Diamond \K 3}$.

 \medskip
 
\textbf{Theorem 6.4} (Tableau Metatheorem) For all $X \subseteq \For$, $A \in \For$ the following statements are equivalent:
\begin{itemize}

\item $X \models_{\Diamond \K 3} A$

\item  $X \vartriangleright_{\Diamond \K 3} A$

\item there is a finite $Y \subseteq X$ and a closed tableau $\langle Y, A, \Phi \rangle$.
\end{itemize}

 Thus, the metatheory presented in the paper simplified the process of defining all notions and proving specific facts when constructing the $\vartriangleright_{\K \Diamond 3}$ equivalent of a given logic $\models_{\K \Diamond 3}$. In fact, it is enough to define the rules and check their relationship with the models, so in the presented methodology, as we shown, the construction of an adequate tableau system really came down to three steps.\end{myfont} 
\qed{}
\end{ex}

One may however ask the question whether within this methodology we can automate it `more and more widely'? 

`More' here means that given $\models$ (the class of structures and the truth conditions), we could simply automatically generate a set of tableau rules that specify the appropriate $\vartriangleright$ (so no proof is needed that the models and rules are sound to each other). This seems to be possible to some extent, but would require further examination of the metatheory presented. 

On the other hand, the term `more widely' means that we could establish a similar methodology for a first-order language, and perhaps a universal approach, covering all kinds of languages in which logical systems are formulated. This also seems possible. What the theory covers turns out to be all sufficient features of any tableau system defined by general semantic structures. In the standard approach --- unlike the one presented here --- it seems that it is very difficult to prove general facts about classes of logics, because we do not have universal and precise concepts that are constant and differ only from one set of tableau rules to another. So the meta-theory can of course still be generalized and its range of applications can be much wider. Of course, the ultimate goal of working on a metatheory of tableau should be to achieve the same level of generality as we find in the metatheory of axiomatic systems.

We believe that having such a general methodology for tableau systems is a desirable condition for general study of the problem, such as: estimating proofs, economics of formulating tableau systems, reducing labeled to unlabeled systems, comparing different tableau systems in one theoretical approach, comparing our approach to other approaches/paradigms --- in terms of the level of generality, automation of the creation of appropriate tableau systems, the scope of automatic transformations of tableau systems into other types of deductive systems. These are future challenges that require a theoretical framework, which we have partially developed here, so that the logic can better address them. 

\section*{References}\label{sec:8}

\begingroup
\renewcommand{\section}[2]{}%
\bibliography{bib}

@incollection{EB1955,
        author = "Evert W. Beth",
	title = "Semantic Entailment and Formal Derivability",
        editor =   " ",
        booktitle   = "Semantic Entailment and Formal Derivability. Koninklijke Nederlandse Akademie van Wentenschappen, Proceedings of the Section of Sciences, 18",
        publisher = "Noord-Hollandsche",
        year = "1955",
	  pages       = "309--342",
        remark =  {Reprinted in Jaakko Hintikka (ed.) The Philosophy of Mathematics, Oxford University Press, 1969}
}

@article{ZL1960,
  author =       "Zbigniew Lis",
  title =        "Wynikanie semantyczne a wynikanie formalne (semantic entaliment vs. syntactic entailment)",
  journal =      "Studia Logica",
  volume =       "10",
  pages =        "39--54",
  year =         "1960",
  DOI =          "https://doi.org/10.1007/BF02120399"
}

@article{RRiVR1972,
    author     =   "Routely Richard and Valerie Routley",
    title      =   "The Semantics of First Degree Entailment",
    journal    =   "Noûs",
    volume     =   "6(4)",
    year       =   "1972",
    pages      =   "335--369",
}

@book{MF1983, 

        title = "Proof Methods for Modal and Intuitionistic Logics",
        year = "1983",
        author = "Melvin Fitting",
        publisher = "D. Reidel, Dordrecht, Holland" 
}

@book{RE1990,
    author    = "Richard L. Epstein",
    title     = "The semantic foundations of logic. Volume 1: Propositional logics",
    year      = "1990",
    publisher = "Springer Science+Business Media",
    address   = "Dordrecht"
}

@incollection{RG1999,
        author = "Rajeev Gor\'{e}",
	title = "Tableau Methods for Modal and Temporal Logics",
        editor =   "D'Agostino M., Gabbay D., Haehnle R. and Posegga J.",
        booktitle   = "Handbook of Tableau Methods",
        year = "1999",
        publisher = "Kluwer",
	  pages     = "297--396"
}

@article{TJ2007, 
    author     =   "Tomasz Jarmu\.{z}ek",
    title      =   "Construction of tableaux for classical logic: tableaux as combinations of branches, branches as chains of sets",
    journal    =   "Logic and Logical Philosophy",
    volume     =   "1(16)",
    year       =   "2007",
    pages      =   "85--101"
}

@article{TJ2008, 
    author     =   "Tomasz Jarmu\.{z}ek",
    title      =   "Tableau system for logic of categorial propositions and decidability",
    journal    =   "Bulletin of the Section of Logic",
    volume     =   "37",
    year       =   "2008",
    pages      =   "223--231"
}

@book{GP2008,
    author      =   "Graham Priest",
    year        =   "2008",
    title       =   "An introduction to non-classical  logic",
    publisher   =   "Cambridge University Press"
}

@incollection{TJ2013,
        author      =  "Tomasz Jarmu\.{z}ek",
	    title       = "Tableau Metatheorem for Modal Logics",
        editor      =   "Ciuni R., Wansing H. and Willkomennen C.",
        booktitle   = "Recent Trends in Philosphical Logic, Trends in Logic",
        year        = "2013",
        publisher   = "Springer Verlag",
	    pages       = "105--128"
}

@book{TJ2013b,
    author      =   "Tomasz Jarmu\.{z}ek ", 
    year        =   "2013", 
    title       =   "Formalizacja metod tablicowych dla logik zda\'{n} i logik nazw (Formalization of tableau methods for propositional logics and for logics of names)", 
    publisher   =   "Wydawnictwo UMK, Toru\'{n}"
}

@incollection{TJiMT2013,
        author =  "Tomasz Jarmu\.{z}ek and Tkaczyk M.",
	    title = "A method of defining paraconsistent tableaux",
        editor =   "J. Y. Beziau, M. Chakraborty and S. Dutta",
        booktitle   = "New Directions in Paraconsistent Logic,  vol. 152 of series Springer Proceedings in Mathematics and Statistics",
        year = "2015",
        publisher = "Springer Indie",
	  pages     = "296--307"
}

@article{TJiMT2015,
    author =  "Tomasz Jarmu\.{z}ek and Tkaczyk M.",
    title      =   "Modal paraconsistent tableau systems of logic",
    journal    =   "WSEAS Transactions on Mathematics",
    volume     =   "14",
    year       =   "2015",
    pages      =   "248--255"
}

@article{DM2017,
  author =       "David C. Makinson",
  title =        "Relevance via decomposition: A project, some results, an open question",
  journal =      "Australasian Journal of Philosophy",
  volume =       "14",
  number =       "3",
  pages =        "356--377",
  year =         "2017",
  DOI =          ""
}

@article{TJiMK2020,
  author    = "Tomasz Jarmu{\.z}ek and Mateusz Klonowski",
  title     = "On logics of strictly-deontic modalities. {A} semantic and tableau approach",
  journal   = "Logic and Logical Philosophy",
  volume    = "29",
  number    = "3",
  pages     = "335--380",
  year      = "2020",
  DOI       = "10.12775/LLP.2020.010"
}

@incollection{TJiRG2021,
  author      = "Tomasz Jarmu{\.z}ek and Rajeev Gor{\'e}",
  title       = "Tableau metatheory for syllogistic logics",
  editor      = "M. Fitting",
  booktitle   = "Selected topics from contemporary logics",
  publisher   = "College Publications",
  year        = "2021",
  pages       = "539--582"
}

@article{TJiMK2022,
    author      =   "Tomasz Jarmu{\.z}ek and Mateusz Klonowski",
    year        =   "2022",
    title       =   "Tableaux for logics of content relationship and set-assignment semantics",
    journal     =   "Logica Universalis",
    volume      =   "16",
    pages       =   "195--219",
    url         =   "https://doi.org/10.1007/s11787-022-00293-w",
    DOI         =   "10.1007/s11787-022-00293-w"
}
\bibliographystyle{plain}
\endgroup





\end{document}